%% file: ggc-ls.tex
\newif\iflong\longtrue
\newif\iflonglong\longlongtrue

\documentclass{article}

\usepackage{authblk}

\usepackage[disable]{todonotes}

\usepackage{times}
\usepackage{latexsym}
\usepackage{soul}
\usepackage{url}
\usepackage[hidelinks]{hyperref}
\usepackage[utf8]{inputenc}
\usepackage[small]{caption}
\usepackage{graphicx}
\usepackage{amsmath}
\usepackage{amsthm}
\usepackage{booktabs}
\usepackage[switch]{lineno}

\usepackage{amssymb}
\usepackage[arrow, matrix, curve]{xy}

\usepackage{graphicx}
\usepackage{dsfont}
\usepackage{braket}
\usepackage{enumitem}
\usepackage{tikz}
\usepackage{verbatim}
\usepackage{mathrsfs} 
\usepackage{xspace}
\usepackage{mathtools}
\usepackage{nicefrac}
\usepackage{cleveref}
\usepackage[graphicx]{realboxes}
\usepackage[noend]{algpseudocode}
\usetikzlibrary{positioning,chains,fit,shapes,calc}

\author[1]{Jaroslav Garvardt}
\author[2]{Niels Grüttemeier}
\author[1]{Christian Komusiewicz}
\author[1]{Nils Morawietz}
\affil[1]{Friedrich Schiller University Jena, Institute of Computer Science, Germany}
\affil[2]{Fraunhofer Institute of Optronics, System Technologies and Image Exploitation, Germany}
\date{\{jaroslav.garvardt, c.komusiewicz, nils.morawietz\}@uni-jena.de \\
niels.gruettemeier@iosb-ina.fraunhofer.de
}

\newcommand{\dflip}{d_\mathrm{flip}}
\newcommand{\Dflip}{D_\mathrm{flip}}

\newtheorem{definition}{Definition}
\newtheorem{proposition}{Proposition}
\newtheorem{claim}{Claim}
\newtheorem{theorem}{Theorem}
\newtheorem{lemma}{Lemma}
\newtheorem{corollary}{Corollary}

\newcommand{\Oh}{\ensuremath{\mathcal{O}}}

\newenvironment{claimproof}{{\noindent\textit{Proof.}}}{\hfill$\blacksquare$}

\newcommand{\todom}[2][]{\todo[#1,color=green!50]{#2}}
\newcommand{\todomi}[1]{\todom[inline]{#1}}

\definecolor{myred}{rgb}{1,0.25,0.25}

\newcommand{\pF}[3]{\ensuremath{#1_{#3}}}

\newcommand{\prob}[3]{\begin{quote}  \textsc{#1}\\  \textbf{Input:} #2\\  \textbf{Question:} #3\end{quote}}

\newcommand{\W}[1]{\ensuremath{\mathrm{W}[#1]}\xspace}
\newcommand\NP{\ensuremath{\mathrm{NP}}\xspace}
\newcommand\Ptime{\ensuremath{\mathrm{P}}\xspace}

\newcommand\FPT{\ensuremath{\mathrm{FPT}}\xspace}

\newcommand{\IS}{\textsc{Independent Set}\xspace}

\newcommand{\GGC}{\textsc{Max~$c$-Cut}\xspace}
\newcommand{\MC}{\textsc{Max Cut}\xspace}

\newcommand{\LGGClong}{\textsc{LS~Max~$c$-Cut}\xspace}

\newcommand{\LGGC}[1][c]{\textsc{LS Max~$#1$-Cut}\xspace}

\newcommand{\LVC}{\textsc{LS-Vertex Cover}\xspace}

\newcommand{\CL}{\textsc{Clique}\xspace}
\newcommand{\LMC}{\textsc{LS Max Cut}\xspace}

\usepackage{ tipa }
\newcommand{\colo}{\chi}

\usetikzlibrary{decorations.pathreplacing, decorations.pathmorphing, calc, shapes.arrows, positioning, matrix}

\include{tabmini}
  \include{tab2}
  \include{tab3}
  \include{tab4}

\title{Parameterized Local Search for Max~$c$-Cut}

\begin{document}
\include{tableImpSmall}

\maketitle

\begin{abstract}
In the NP-hard \textsc{Max~$c$-Cut} problem, one is given an undirected edge-weighted graph~$G$ and aims to color the vertices of~$G$ with~$c$ colors such that the total weight of edges with distinctly colored endpoints is maximal. 
The case with~$c=2$ is the famous \textsc{Max Cut} problem.  
To deal with the NP-hardness of this problem, we study parameterized local search algorithms. 
More precisely, we study \textsc{LS~Max~$c$-Cut}  where we are also given a vertex coloring and an integer~$k$ and the task is to find a better coloring that changes the color of at most~$k$ vertices, if such a coloring exists; otherwise, the given coloring is $k$-optimal. 
We show that, for all~$c\ge 2$, \textsc{LS~Max~$c$-Cut} presumably cannot be solved in~$f(k)\cdot n^{\Oh(1)}$~time even on bipartite graphs. 
We then present an algorithm for \textsc{LS~Max~$c$-Cut} with running time~$\Oh((3e\Delta)^k\cdot c\cdot k^3\cdot\Delta\cdot n)$, where~$\Delta$ is the maximum degree of the input graph.
Finally, we evaluate the practical performance of this algorithm in a hill-climbing approach as a post-processing for a state-of-the-art heuristic for~\textsc{Max~$c$-Cut}. 
We show that using parameterized local search, the results of this state-of-the-art heuristic can be further improved on a set of standard benchmark instances.  
\end{abstract}

\section{Introduction}
Graph coloring and its generalizations are among the most famous NP-hard optimization problems with numerous practical applications~\cite{JT11}. 
In one prominent problem variant, we want to color the vertices of an edge-weighted graph with~$c$ colors so that the sum of the  weights of all edges that have endpoints with different colors is maximized.   
This problem is known as
\textsc{Max $c$-Cut}~\cite{FJ97,KKLP97} or \textsc{Maximum Colorable Subgraph}~\cite{PY91}.
Applications of \textsc{Max $c$-Cut} include data clustering~\cite{CMA21,FKP22}, computation of rankings~\cite{CMA21}, design of experimental studies~\cite{ADR21}, sampling of public opinions in social networks~\cite{HLC17}, channel assignment in wireless networks~\cite{SGDC08}, module detection in genetic interaction data~\cite{LTCH11}, and scheduling of TV commercials~\cite{GKK09}. 
In addition, \textsc{Max $c$-Cut} is closely related to the energy minimization problem in Hopfield neural networks~\cite{SOA99,KT06,Wang06}.
An equivalent formulation of the problem is to ask for a coloring that minimizes the weight sum of the edges whose endpoints receive the same color; this formulation is known as \textsc{Generalized Graph Coloring}~\cite{VL03}.
The main difference is that for instances of~\textsc{Generalized Graph Coloring}, one usually assumes that all edge weights are non-negative, whereas for \textsc{Max $c$-Cut}, one usually also allows negative weights.

From an algorithmic viewpoint, even restricted cases of \textsc{Max~$c$-Cut}
are  hard: 
The special case~$c=2$ is the~\textsc{Max Cut}
problem which is NP-hard even for positive unit weights~\cite{K72,GJ79}, even on graphs with maximum degree 3~\cite{BK99}. 
Moreover, for all~$c\ge 3$ the \textsc{Graph Coloring} problem where we
ask for a coloring of the vertices with $c$~colors such that the endpoints of every edge receive different colors is NP-hard~\cite{K72}. 
As a consequence, \textsc{Max~$c$-Cut} is NP-hard for all~$c\ge 3$, again even when all edges have positive unit weight. 
While \textsc{Max $c$-Cut} admits polynomial-time constant factor approximation algorithms~\cite{FJ97}, there are no polynomial-time approximation schemes unless~$\Ptime = \NP$~\cite{PY91}, even on graphs with bounded maximum degree~\cite{BK99}.
Due to
these hardness results, \textsc{Max~$c$-Cut} is mostly solved using
heuristic approaches~\cite{FPRR02,LTCH11,MH17,ZLA13}.

A popular heuristic approach for  \textsc{Max~$c$-Cut}
is hill-climbing local search~\cite{FPRR02,LTCH11} with the 1-flip neighborhood. 
Here, an
initial solution (usually computed by a greedy algorithm) is replaced by a better one in
the 1-flip neighborhood as long as such a better solution exists. 
Herein, the 1-flip
neighborhood of a coloring is the set of all colorings that can be obtained by
changing the color of one vertex. 
A coloring that has no improving 1-flip is called
1-optimal and the problem of computing 1-optimal solutions has also received interest from
a theoretical standpoint: Finding 1-optimal solutions for \textsc{Max Cut} is PLS-complete
on edge-weighted graphs~\cite{SY91} and thus presumably not efficiently solvable in
the worst case. 
This PLS-completeness result for the 1-flip neighborhood was later
extended to \textsc{Generalized Graph Coloring}, and thus to \textsc{Max $c$-Cut}, for all~$c$~\cite{VL03}. 
For graphs where the absolute values of all edge weights are constant, however, a simple hill climbing algorithm terminates after~$\Oh(m)$~improvements, where~$m$ is the number of edges in the input graph. 
Here, a different question arises: Can we replace the 1-flip neighborhood with a larger efficiently searchable neighborhood, to avoid being stuck in a bad local optimum?  
A natural candidate is the $k$-flip neighborhood where we are allowed to change the color of at most~$k$ vertices. 
As noted by Kleinberg and Tardos~\cite{KT06}, a standard algorithm for searching the $k$-flip neighborhood takes $\Theta(n^k\cdot m)$~time where~$n$ is the number of vertices. 
This led Kleinberg and Tardos to conclude that the $k$-flip neighborhood is impractical. 
In this work, we ask whether we can do better than the brute-force $\Theta(n^k\cdot m)$-time algorithm or, in other words, whether the dismissal of $k$-flip neighborhood may have been premature.

The ideal framework to answer this question is parameterized local search, where the
main goal would be to design an algorithm that in~$f(k)\cdot n^{\Oh(1)}$~time either
finds a better solution in the~$k$-flip neighborhood or correctly answers that the current
solution is~$k$-optimal. 
Such a running time is preferable to~$\Oh(n^k\cdot m)$ since the degree of the polynomial running time part does not depend on~$k$ and thus the running time scales better with~$n$. 
The framework also provides a toolkit for negative results that allows to
conclude that an algorithm with such a running time is unlikely by showing
\W1-hardness. 
In fact, most parameterized local search problems turn out to be \W1-hard
with respect to the search radius~$k$~\cite{BIJK19,DGKW14,FFL+12,GHNS13,GHK14,GKO+12,KLMS23,Marx08,Szei11}.  
In contrast to these many, mostly negative, theoretical results, there are so far only few encouraging experimental studies~\cite{GGJ+19,GKM21,HN13,KK17}. 
The maybe most extensive positive results so far were obtained for \textsc{LS Vertex Cover} where the input is an undirected graph~$G$ with a vertex cover~$S$ and the question is whether the~$k$-swap neighborhood\footnote{The~$k$-swap neighborhood of a vertex cover~$S$ of~$G$ is the set of all vertex covers of~$G$ that have a symmetric difference of at most~$k$ with~$S$.} of~$S$ contains a smaller vertex cover. 
The key to obtain practical parameterized local search algorithms is to consider parameterization by~$k$ \emph{and} the maximum degree~$\Delta$ of the input graph. 
As shown by Katzmann~and~Komusiewicz~\cite{KK17}, \textsc{LS Vertex Cover} can be solved in~$(2\Delta)^k\cdot n^{\Oh(1)}$~time.
An experimental evaluation of this algorithm showed that it can be tuned to solve the problem for~$k\approx 20$ on large sparse graphs, and that~$k$-optimal solutions for~$k\ge 9$ turned out to be optimal for almost all graphs considered in the experiments.

\subparagraph{Our Results.} We study \LGGClong, where we want to decide whether a given coloring has a better one in its~$k$-flip neighborhood. 
We first show that \LGGClong{} is presumably not solvable in $f(k)\cdot n^{o(k)}$~time and transfer this lower-bound also to local search versions of related partition problems like~\emph{Min Bisection} and~\emph{Max Sat}.
We then present an algorithm to solve~\LGGClong in time~$\Oh((3e\Delta)^k\cdot c \cdot  k^3\cdot\Delta\cdot n)$, where~$\Delta$ is the maximum degree of the input graph. 
To put this running time bound into context, two aspects should be discussed:
First, the NP-hardness of the special case of \GGC{} with~$\Delta=3$ implies that a running time of~$f(\Delta)\cdot n^{\Oh(1)}$ is impossible unless~$\Ptime = \NP$. 
Second, only the parameter~$k$ occurs in the exponent; we say that the running time grows mildly with respect to~$\Delta$ and strongly with respect to~$k$. 
This is desirable as~$k$ is a user-determined parameter whereas~$\Delta$ depends on the input; a broader discussion of this type of running times is given by Komusiewicz and Morawietz~\cite{KM22}.

The algorithm is based on two main observations: 
First, we show that minimal improving flips are connected in the input graph. 
This allows to enumerate candidate flips in~$\Oh((e\Delta)^k \cdot k\cdot n)$~time. 
Second, we show that, given a set of~$k$ vertices to flip, we can determine an optimal way to flip their colors in $\Oh(3^k\cdot c\cdot k^2+k\cdot \Delta)$~time. 
We then discuss several ways to speed up the algorithm, for example by computing upper bounds for the improvement of partial flips. 
We finally evaluate our algorithm experimentally when it is applied as post-processing for a state-of-the-art \textsc{Max $c$-Cut}~heuristic~\cite{MH17}. 
In this application, we take the solutions computed by the heuristic and improve them by hill-climbing with the~$k$-flip neighborhood for increasing values of~$k$. 
We show that, for a standard benchmark data set, a large fraction of the previously best solutions can be improved by our algorithm, leading to new record solutions for these instances. 
The post-processing is particularly successful for the instances of the data set with~$c>2$ and both positive and negative edge weights.

\section{Preliminaries}

\paragraph{Notation.}
For integers~$i$ and~$j$ with~$i \leq j$, we define~$[i,j] := \{k \in \mathds{N}\mid i \leq k \leq j\}$.
For a set~$A$, we denote with~${{A}\choose {2}}:= \{\{a,b\}\mid a \in A, b\in A\}$ the collection of all size-two subsets of~$A$.
For two sets~$A$ and~$B$, we denote with~$A \oplus B := (A \setminus B) \cup (B \setminus A)$ the \emph{symmetric difference} of~$A$ and~$B$.
An~\text{$r$-partition of a set~$C$} is an~$r$-tuple~$(B_1,\dots,B_r)$ of subsets of~$C$, such that each element of~$C$ is contained in exactly one set of~$(B_1,\dots,B_r)$.
For~$r=2$, we may call a~$2$-partition~$(A,B)$ simply a~\emph{partition}.
Let~$f:A\to B$ and~$g:A\to B$ be functions and let~$C\subseteq A$, then we say that~$f$ and~$g$~\emph{agree} on~$S$, if for each element~$s\in S$, $f(s) = g(s)$.

An (undirected) graph~$G=(V,E)$ consists of a vertex set~$V$ and an edge set~$E \subseteq {{V}\choose {2}}$.
For vertex sets~$S\subseteq V$ and~$T\subseteq V$ we denote with~$E_G(S,T) := \{\{s,t\}\in E \mid s\in S, t\in T\}$ the edges between~$S$ and~$T$ and with~$E_G(S) := E_G(S,S)$ the edges between vertices of~$S$.
Moreover, we define~$G[S] := (S,E_G(S))$ as the~\emph{subgraph of~$G$ induced by~$S$}. 
A vertex set~$S$ is \emph{connnected} if~$G[S]$ is a connected graph.
For a vertex~$v\in V$, we denote with~$N_G(v):= \{w\in V\mid \{v,w\}\in E\}$ the \emph{open neighborhood} of~$v$ in~$G$ and with~$N_G[v]:= N_G(v) \cup \{v\}$ the \emph{closed neighborhood} of~$v$ in~$G$. 
Analogously, for a vertex set~$S\subseteq V$, we define~$N_G[S] := \bigcup_{v\in S} N_G[v]$ and~$N_G(S) := \bigcup_{v\in S} N_G(v)\setminus S$.
If~$G$ is clear from context, we may omit the subscript.
We say that~\emph{vertices~$v$ and~$w$ have distance at least~$i$} if the length of the shortest path between~$v$ and~$w$ is at least~$i$.

\iflonglong
\fi
\paragraph{Problem Formulation.}

Let~$X$ and~$Y$ be sets and let~$\colo,\colo': X\to Y$.
The~\emph{flip} between~$\colo$ and~$\colo'$ is defined as~$\Dflip(\colo,\colo'):=\{x\in X\mid \colo(x) \neq \colo'(x)\}$ and the~\emph{flip distance} between~$\colo$ and~$\colo'$ is defined as~$\dflip(\colo,\colo'):=|\Dflip(\colo,\colo')|$.
For an integer~$c$ and a graph~$G=(V,E)$, a function~$\colo:V\to[1,c]$ is a~\emph{$c$-coloring} of~$G$.
Let~$\colo$ be a $c$-coloring of~$G$, we define the set~$E(\colo)$ of~\emph{properly colored edges} as~$E(\colo):=\{\{u,v\}\in E\mid \colo(u)\neq \colo(v)\}$.  
For an edge-weight function~$\omega:E\to \mathds{Q}$ and an edge set~$E'\subseteq E$, we let~$\omega(E')$ denote the total weight of all edges in~$E'$.
Let~$\colo$ and~$\colo'$ be~$c$-colorings of~$G$.
We say that~$\colo$ and~$\colo'$ are~\emph{$k$-neighbors} if~$\dflip(\colo,\colo') \leq k$. 
If~$\omega(E(\colo)) > \omega(E(\colo'))$, we say that~$\colo$ is~\emph{improving} over~$\colo'$.
Finally, a~$c$-coloring~$\colo$ is~\emph{$k$-(flip-)optimal} if~$\colo$ has no improving~$k$-neighbor~$\colo'$.
%
The problem of finding an improving neighbor of a given coloring can now be formalized as follows.  
\prob{\LGGC}{A graph~$G=(V,E)$,~$c\in\mathds{N} $, a weight function~$\omega: E \to \mathds{Q}$, a~$c$-coloring~$\colo$, and~$k\in \mathds{N}$.}{Is there a~$c$-coloring~$\colo'$ such that~$\dflip(\colo,\colo')\leq k$ and~$\omega(E(\colo')) > \omega(E(\colo))$?}

\iflonglong

\fi

\iflong
The special case of~\LGGC where~$c=2$ can alternatively be defined as follows by using partitions instead of colorings.

\prob{\LMC}{A graph~$G=(V,E)$, a weight function~$\omega: E \to \mathds{Q}$, a partition~$(A,B)$ of~$V$, and~$k\in \mathds{N}$.}{Is there a set~$S\subseteq V$ of size at most~$k$ such that~$\omega(E(A,B)) < \omega(E(A\oplus S,B\oplus S))$?}
\fi
While these problems are defined as decision problems, our algorithms solve the search problem that returns an improving $k$-flip if it exists.

Let~$\colo$ and~$\colo'$ be~$c$-colorings of a graph~$G$.
We say that~$\colo'$ is an~\emph{inclusion-minimal improving $k$-flip for~$\colo$}, if~$\colo'$ is an improving~$k$-neighbor of~$\colo$ and if there is no improving~$k$-neighbor~$\widetilde{\colo}$ of~$\colo$ with~$\Dflip(\colo,\widetilde{\colo}) \subsetneq \Dflip(\colo,\colo')$.
\iflong Let~$(A,B)$ be a partition of~$G$.
In the context of~\LMC, we call a vertex set~$S$~\emph{inclusion-minimal improving $k$-flip for~$(A,B)$}, if~$|S|\leq k$, $\omega(E(A\oplus S, B\oplus S)) > \omega(E(A,B))$, and if there is no vertex set~$S' \subsetneq S$ such that~$\omega(E(A\oplus S', B\oplus S')) > \omega(E(A,B))$.
\fi

\iflong

\fi

In this work, we also consider the \emph{permissive version} of the above local search problems.
In such a permissive version~\cite{GKO+12}, we get the same input as in the normal local search problem, but the task is now to (i)~find \emph{any} better solution or (ii)~correctly output that there is no better solution in the $k$-neighborhood.

\section{W[1]-hardness and a tight ETH lower bound for~\LGGC and related problems}\label{sec:ggc hardness}
We first show our intractability result for~\LMC. 
More precisely, we show that \LMC is~\W1-hard when parameterized by~$k$ even on bipartite graphs with unit weights. 
This implies that even on instances where an optimal partition can be found in linear time, \LMC~presumably cannot be solved within~$f(k) \cdot n^{\Oh(1)}$ time for any computable function~$f$.
Afterwards, extend the intractability results even to the permissive version of~\LGGC on general graphs.
Finally, we can then also derive new intractability results for local search versions for the related partition problems~\textsc{Min Bisection}, \textsc{Max Bisection}, and~\textsc{Max Sat}.

To prove the intractability results for the strict version, we introduce the term of \emph{blocked vertices} in instances with unit weights. 
Intuitively, a vertex~$v$ is blocked for a color class~$i$ if we can conclude that~$v$ does not move to~$i$ in any optimal~$k$-neighbor of the current solution just by considering the graph neighborhood of~$v$. 
This concept is formalized as follows.

\begin{definition}\label{def:blocker}
Let~$G=(V,E)$ be a graph, let~$\colo$ be a~$c$-coloring of~$G$, and let~$k$ be an integer.
Moreover, let~$v$ be a vertex of~$V$ and let~$i\in[1,c]\setminus \{\colo(v)\}$ be a color.
The vertex~$v$ is~\emph{$(i,k)$-blocked in~$G$ with respect to~$\colo$} if~$v$ has at least~$2k+1$ more neighbors of color~$i$ than of color~$\colo(v)$ with respect to~$\colo$, that is, if~$|\{w\in N(v) \mid \colo(w) = i\}| \geq |\{w\in N(v) \mid \colo(w) = \colo(v)\}| + 2k - 1$.
\end{definition}
Note that a partition~$P := (B_1,B_2)$ can be interpreted as the~$2$-coloring~$\colo_{P}$ defined for each vertex~$v\in V$ by~$\colo_{P}(v) := i$, where~$i$ is the the unique index of~$\{1,2\}$ such that~$v\in B_i$. 
Hence, we may also say that a vertex~$v$ is~\emph{$(B_i,k)$-blocked in~$G$ with respect to~$(B_1,B_2)$}, if~$v$ is $(i,k)$-blocked in~$G$ with respect to~$\colo_P$.

\subsection{Hardness for the strict version of~\LGGC}

\begin{lemma}\label{lem:solutionwithoutblocked}
Let~$G=(V,E)$ be a graph, let~$\colo$ be a~$c$-coloring of~$G$, let~$k$ be an integer. 
Moreover, let~$v$ be a vertex in~$V$ which is~$(i,k)$-blocked in~$G$ with respect to~$\colo$.
Then, there is no inclusion-minimal improving~$k$-neighbor~$\colo'$ of~$\colo$ with~$\colo'(v) = i$.
\end{lemma} 
\iflong
\begin{proof}
Let~$\colo'$ be a~$c$-coloring of~$G$ with~$\dflip(\colo,\colo')\leq k$ and~$\colo'(v) = i$.
Hence,~$v\in \Dflip(\colo,\colo')$ and thus~$\Dflip(\colo,\colo')$ contains at most~$k-1$ neighbors of~$v$.
Consequently, at most~$k-1$ more neighbors of~$v$ receive color~$\colo(v)$ under~$\colo'$ than under~$\colo$.
Similarly, at most~$k-1$ more neighbors of~$v$ receive color~$i$ under~$\colo$ than under~$\colo'$.
Since~$v$ is~$(i,k)$-blocked in~$G$ with respect to~$\colo$, this then implies that~$v$ has more neighbors of color~$i$ than of color~$\colo(v)$ under~$\colo'$.
Let~$\colo^*$ be the~$c$-coloring of~$G$ that agrees with~$\colo'$ on all vertices of~$V\setminus\{v\}$ and where~$\colo^*(v):=\colo(v)$.
Note that~$E(\colo')\setminus E(\colo^*) = \{\{w,v\}\in E\mid \colo'(w) = \colo(v)\}$ and~$E(\colo^*)\setminus E(\colo') = \{\{w,v\}\in E\mid \colo'(w) = i\}$.
This implies that~$\colo^*$ is a better~$c$-coloring for~$G$ than~$\colo'$, since~$$|E(\colo^*)|-|E(\colo')| =  |E(\colo^*)\setminus E(\colo')| - |E(\colo')\setminus E(\colo^*)| > 0.$$
Hence, $\colo'$ is not an inclusion-minimal improving~$k$-neighbor of~$\colo$, since~$\Dflip(\colo,\colo^*) = \Dflip(\colo,\colo')\setminus\{v\} \subsetneq \Dflip(\colo,\colo')$.
\end{proof}
\fi
The idea of blocking a vertex by its neighbors finds application in the construction for the \W1-hardness from the next theorem.

\begin{theorem}\label{thm:maxcutWhard}
\LMC is \W1-hard \iflong{}when parameterized by\fi{}~$k$ on bipartite~$2$-degenerate graphs with unit weights.
\end{theorem}
\begin{proof}
We reduce from~\CL, where we are given an undirected graph~$G$ and~$k\in\mathds{N}$ and ask whether~$G$ contains a clique of size~$k$. 
\CL~is~\W1-hard \iflong when parameterized by \fi the size~$k$ of the sought clique\iflong~\cite{DF13,C+15}\fi.


Let~$I:=(G=(V,E),k)$ be an instance of~\CL.
\iflong{}We construct an equivalent instance~$I':=(G'=(V',E'),\omega',A',B',k')$ of~\LMC with $\omega': E' \to \{1\}$ as follows.
\fi
We start with an empty graph~$G'$ and add each vertex of~$V$ to~$G'$.
Next, for each edge~$e\in E$, we add two vertices~$u_e$ and~$w_e$ to~$G'$ and make both~$u_e$ and~$w_e$ adjacent to each endpoint of~$e$ in~$G'$.
Afterwards, we add a vertex~$v^*$ to~$G'$ and for each edge~$e\in E$, we add vertices~$x_e$ and~$y_e$ and edges~$\{w_e, x_e\},\{w_e, y_e\}$, and~$\{x_e, v^*\}$ to~$G'$.
Finally, we add a set~$V_z$ of~$|E|-2\cdot \binom{k}{2} + 1$ vertices to~$G'$ and make each vertex of~$V_z$ adjacent to~$v^*$.

In the following, for each~$\alpha\in \{u,w,x,y\}$, let~$V_\alpha$ denote the set of all~$\alpha$-vertices in~$G'$, that is, $V_\alpha := \{\alpha_e\mid e\in E\}$. 
We set
\begin{align*}
B'&:= V_w\cup \{v^*\}\cup V_z\text{, } A' := V' \setminus B'\text{, and}\\
k' &:= 2\cdot \binom{k}{2} + k + 1.
\end{align*}

To ensure that some vertices are blocked in the final graph~$G'$, we add the following further vertices to~$A'$ and~$B'$:
For each vertex~$v'\in V_u \cup V_y$, we add a set of~$2k'+2$ vertices to~$B'$ that are only adjacent to~$v'$ and for each vertex~$v'\in V_z$, we add a set of~$2k'+2$ vertices to~$A'$ that are only adjacent to~$v'$.
Let~$V_\Gamma$ be the set of those additional~vertices. Figure~\ref{fig max cut hardness} shows a sketch of the vertex sets and their connections in~$G'$. 
Note that~$G'$ is bipartite and~$2$-degenerate.\todomi{mention degeneracy ordering}
 
\begin{figure}[t]
\begin{center}
\begin{tikzpicture}[yscale=.6, xscale = 1.5]
\tikzstyle{knoten}=[circle,fill=white,draw=black,minimum size=7pt,inner sep=0pt]
\tikzstyle{blocked}=[rectangle,fill=white,draw=black,minimum size=7pt,inner sep=0pt]
\tikzstyle{bez}=[inner sep=0pt]

		\node (a)[label=left:{$A'$}] at (-2.2,-.2) {};
\draw[rounded corners, fill=yellow!30] (-2.2, 0.6) rectangle (2.7, -1.3) {};
		\node[knoten] (v)[label=below:{$V$}] at (-2,0) {};
		\node[blocked] (vu)[label=below:{$V_u$}] at (-1,0) {};
		\node[blocked] (vy)[label=below:{$V_y$}] at (0,0) {};
		\node[knoten] (vx)[label=below:{$V_x$}] at (1,0) {};

\begin{scope}[yshift=-1cm]
		\node (b)[label=left:{$B'$}] at (-2.2,4.2) {};
\draw[rounded corners, fill=blue!30] (-2.2, 5.2) rectangle (2.7, 3.3) {};

		\node[knoten] (vw)[label={$V_w$}] at (-.5,4) {};	
		\node[knoten] (vc)[label={$v^*$}] at (1.5,4) {};
		\node[blocked] (vz)[label={$V_z$}] at (2.5,4) {};    
\end{scope}

		\draw[-, ultra thick] (v) to (vu);
		\draw[-, ultra thick] (v) to (vw);
		\draw[-, ultra thick] (vw) to (vy);
		\draw[-, ultra thick] (vw) to (vx);
		\draw[-, ultra thick] (vx) to (vc);
		\draw[-, ultra thick] (vc) to (vz);           
		\end{tikzpicture}
\end{center}
\caption{The connections between the different vertex sets in~$G'$.
Two vertex sets~$X$ and~$Y$ are adjacent in the figure if~$E(X,Y) \neq \emptyset$.
Each vertex~$v$ in a vertex set with a rectangular node is~$k'$-blocked from the opposite part of the partition.
The vertex set~$V_\Gamma$ is not shown.}
\label{fig max cut hardness}
\end{figure}
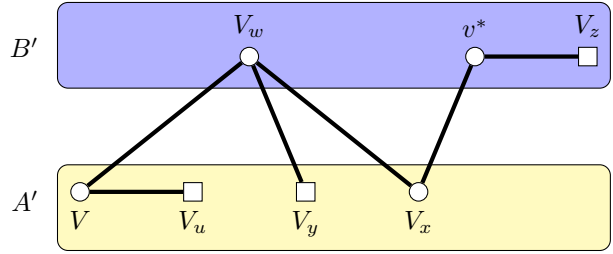

Note that each vertex in~$V_u\cup V_y$ is contained in~$A'$, has at most two neighbors in~$A'$, and at least~$2k'+2$ neighbors in~$B'$. 
Morever, each vertex in~$V_z$ is contained in~$B'$, has one neighbor in~$B'$, and~$2k'+2$ neighbors in~$A'$.
Hence, each vertex in~$V_u \cup V_y$ is~$(B',k')$-blocked and each vertex in~$V_z$ is~$(A',k')$-blocked.
Consequently, due to~\Cref{lem:solutionwithoutblocked}, no inclusion-minimal improving~$k'$-flip for~$(A',B')$ contains any vertex of~$V_u\cup V_y\cup V_z$.
As a consequence, no inclusion-minimal improving~$k'$-flip for~$(A',B')$ contains any vertex of~$V_\Gamma$. 
In other words, only vertices in~$V$, $V_w$, $V_x$, and the vertex~$v^*$ can flip their colors.
 
Before we show the correctness, we provide some intuition. 
By the above, intuitively, a clique~$S$ in the graph~$G$ then corresponds to a flip of vertex~$v^*$, the vertices of~$S$, and the vertices~$w_e$ and~$x_e$ for each edge~$e$ of the clique. 
 The key mechanism is that each inclusion-minimal improving flip has to contain~$v^*$, so that edges between~$v^*$ and~$V_z$ become properly colored.
 To compensate for the edges between~$V_x$ and~$v^*$ that are not properly colored after flipping~$v^*$, for some edges~$e$ of~$G$, the corresponding vertices of~$V_x$ and~$V_w$  and both endpoints of~$e$ have to flip their color.
 The size of~$V_z$ ensures that this has to be done for at least~$\binom{k}{2}$ such edges of~$G$.
Since we only allow a flip of size~$k'$, this then ensures that the edges of~$G$ whose corresponding vertices flip their color belong to a clique of size~$k$ in~$G$.
 %

Next, we show that~$I$ is a yes-instance of~\CL if and only if~$I'$ is a yes-instance of~\LMC.

$(\Rightarrow)$ 
Let~$S\subseteq V$ be a clique of size~$k$ in~$G$.
Hence, $\binom{S}{2}\subseteq E$.
We set~$S' := S \cup \{w_e, x_e \mid e\in \binom{S}{2}\} \cup \{v^*\}$.
Note that~$S'$ has size~$k + 2\cdot \binom{k}{2} + 1 = k'$.
Let~$C := E(A', B')$ and let~$C' := E(A' \oplus S', B'\oplus S')$.
It remains to show that~$C'$ contains more edges than~$C$.
To this end, note that~$C$ and~$C'$ differ only on edges that have at least one endpoint in~$S'$.

First, we discuss the edges incident with at least one vertex of~$S = S'\cap V$.
For each vertex~$v\in S$ and each vertex~$v'\in N_G(v)\setminus S$, the edge~$\{v,w_{\{v,v'\}}\}$ is contained in~$C$ but not in~$C'$ and the edge~$\{v,u_{\{v,v'\}}\}$ is contained in~$C'$ but not in~$C$. 
For each other neighbor~$v'\in N_G(v)\cap S$, the edge~$\{v,u_{\{v,v'\}}\}$ is contained in~$C'$ but not in~$C$ and the edge~$\{v,w_{\{v,v'\}}\}$ is contained in both~$C$ and~$C'$.
Next, we discuss the remaining edges incident with some vertex of~$\{w_e, x_e \mid e\in \binom{S}{2}\}$.
For each edge~$e\in \binom{S}{2}\subseteq E$, the edges~$\{w_e,x_e\}$ and~$\{x_e,v^*\}$ are contained in both~$C$ and~$C'$ and the edge~$\{w_e,y_e\}$ is contained in~$C$ but not in~$C'$.
Finally, we discuss the remaining edges incident with~$v^*$.
For each edge~$e\in E\setminus \binom{S}{2}$, the edge~$\{v^*, x_e\}$ is contained in~$C$ but not in~$C'$ and for each vertex~$z\in V_z$, the edge~$\{v^*, z\}$ is contained in~$C'$ but not in~$C$.
Hence,
\begin{align*}
C\setminus C' =~ &\{\{v, w_{\{v,v'\}}\}\mid v\in S, v'\in N_G(v)\setminus S\}\\
& \cup \{\{w_e,y_e\}\mid e\in\binom{S}{2}\}\\
&\cup \{\{v^*,x_e\}\mid e\in E\setminus \binom{S}{2}\}.
\end{align*}
Furthermore, we have
\begin{align*}
C'\setminus C = \{\{v,u_{\{v,v'\}}\}\mid v\in S, v'\in N_G(v)\} \cup \{\{v^*, z\}\mid z\in V_z\}.
\end{align*}
Since~$|V_z| = |E|- 2\cdot \binom{k}{2}+1$, we get~
\begin{align*}
|C'\setminus C| - |C\setminus C'|
=\ &|\{\{v,u_{\{v,v'\}}\}\mid v\in S, v'\in N_G(v)\}| \\
&- |\{\{v, w_{\{v,v'\}}\}\mid v\in S, v'\in N_G(v)\setminus S\}|\\
&+ |\{\{v^*, z\}\mid z\in V_z\}| \\
&- |\{\{w_e,y_e\}\mid e\in\binom{S}{2}\}| - |\{\{v^*,x_e\}\mid e\in E\setminus \binom{S}{2}\}|\\
=\ & 2\cdot\binom{k}{2} + |E|- 2\cdot \binom{k}{2}+1 - |E| = 1.
\end{align*}
Consequently, $C'$ contains exactly one edge more than~$C$.
Hence,~$I'$ is a yes-instance of~\LMC.


$(\Leftarrow)$
Let~$S'\subseteq V'$ be an inclusion-minimal improving~$k'$-flip for~$(A',B')$.
Due to~\Cref{lem:solutionwithoutblocked}, we can assume that~$S'\subseteq V \cup V_w \cup V_x \cup \{v^*\}$ since all other vertices of~$V'\setminus V_\Gamma$ are blocked from the opposite part of the partition and for each vertex~$x\in V_\Gamma$, the unique neighbor of~$x$ in~$G'$ is thus not contained in~$S'$.
By construction of~$G'$, each vertex~$v\in V$ is adjacent to~$|N_G(v)|$ vertices of~$A'$ and adjacent to~$|N_G(v)|$ vertices of~$B'$.
Since~$S'$ is inclusion-minimal and contains no vertex of~$\{u_e\mid e\in E\}$, for each vertex~$v\in S'\cap V$, there is at least one edge~$e\in E$ incident with~$v$ in~$G$ such that~$S'$ contains the vertex~$w_e$, as otherwise,  removing~$v$ from~$S'$ still results in an even better partition than~$(A'\oplus S',B'\oplus S')$, that is, 
\begin{align*}
|E(A'\oplus (S'\setminus \{v\}), B'\oplus (S'\setminus \{v\}))|  \geq |E(A'\oplus S', B'\oplus S')| > |E(A',B')|. 
\end{align*}

Moreover, recall that~$B'$ contains all vertices of~$V_w$ and each vertex~$w_e\in V_w$ is adjacent to the four vertices~$\{x_e,y_e\}\cup e$ of~$A'$ and is adjacent to no vertex of~$B'$.
Since~$S'$ is inclusion-minimal and contains no vertex of~$V_y$, for each vertex~$w_e\in S'\cap V_w$, all three vertices of~$\{x_e\}\cup e$ are contained in~$S'$, as otherwise, removing~$w_e$ from~$S'$ does not result in a worse partition than~$(A'\oplus S',B'\oplus S')$, that is, 
\begin{align*}
|E(A'\oplus (S'\setminus \{w_e\}), B'\oplus (S'\setminus \{w_e\}))|\geq |E(A'\oplus S', B'\oplus S')| > |E(A',B')|.
\end{align*}

Furthermore, $A'$ contains all vertices of~$V_x$ and each vertex~$x_e\in V_x$ is adjacent to the vertices~$w_e$ and~$v^*$ in~$B'$ and adjacent to no vertex in~$A'$.
Since~$S'$ is inclusion-minimal, for each vertex~$x_e\in S'\cap V_x$, both~$w_e$ and~$v^*$ are contained in~$S'$, as otherwise, removing~$x_e$ from~$S'$ does not result in a worse partition than~$(A'\oplus S',B'\oplus S')$, that is, 
\begin{align*}
|E(A'\oplus (S'\setminus \{x_e\}), B'\oplus (S'\setminus \{x_e\}))|\geq |E(A'\oplus S', B'\oplus S')|> |E(A',B')|.
\end{align*}
Note that the above statements imply that~$S'$ contains~$v^*$.
This is due to the facts that
\begin{enumerate}[label=\alph*)]
\item $S'$ is non-empty,
\item $S'$ contains only vertices of~$V\cup V_w\cup V_x\cup\{v^*\}$,
\item if~$S'$ contains a vertex of~$V$, then~$S'$ contains a vertex of~$V_w$,
\item if~$S'$ contains a vertex of~$V_w$, then~$S'$ contains a vertex of~$V_x$, and
\item if~$S'$ contains a vertex of~$V_x$, then~$S'$ contains the vertex~$v^*$.
\end{enumerate}

Recall that~$v^*$ is adjacent to the~$|E|$ vertices~$V_x$ in~$A'$ and to the~$|E|-2\cdot \binom{k}{2}+1$ vertices~$V_z$ in~$B'$.
Hence, since~$S'$ is inclusion-minimal and no vertex of~$V_z$ is contained in~$S'$, $S'$ contains at least~$\binom{k}{2}$ vertices of~$V_x$, as otherwise,
\begin{align*}
|E(A'\oplus (S'\setminus \{v^*\}), B'\oplus (S'\setminus \{v^*\}))| \geq |E(A'\oplus S', B'\oplus S')|> |E(A',B')|.
\end{align*}

Concluding, $S'$ contains~$v^*$ and for at least~$\binom{k}{2}$ edges~$e\in E$ the vertices~$x_e$, $w_e$, and the endpoints of~$e$.
Let~$S:= S'\cap V$.
Since~$S'$ has size at most~$k' = 2\cdot \binom{k}{2} + k + 1$, the above implies that~$S$ has size at most~$k$.
Since~$S'$ contains the endpoints of at least~$\binom{k}{2}$ edges~$e\in E$, we conclude that~$S$ is a clique of size~$k$ in~$G$.
Hence, $I$ is a yes-instance of~\CL.  
\end{proof}

This implies that even on instances where an optimal solution can be found in polynomial time, local optimality cannot be verified efficiently.
This property was was also shown for~\LVC.
Namely, \LVC was shown to be~\W1-hard with respect to the search radius even on bipartite graphs~\cite{GKO+12}.

\iflong
Next, we describe how to adapt the above reduction can to prove~\W1-hardness of \LGGC for each fixed~$c \geq 2$ when parameterized by~$k$. 

Consider the instance~$I:=(G,\omega, (A,B),k)$ of~\LMC that has been constructed in the proof of~\Cref{thm:maxcutWhard} and let~$c > 2$.
For every vertex~$v$ of~$G$, we add further degree-one neighbors. More precisely, for every color~$i \in [3,c]$, the vertex $v$ receives additional neighbors of color~$i$ such that~$v$ is~$(i,k)$-blocked. Let~$G'$ be the resulting graph.

Then, for any inclusion-minimal improving~$k$-flip~$\colo'$ for~$\colo$ of~$G'$, we have~$S := \Dflip(\colo,\colo')\subseteq A\cup B$,~$\colo'(a) = 2$ for each~$a\in A \cap S$, and~$\colo'(b) = 1$ for each~$b\in B \cap S$.
Hence,~$I$ is a yes-instance of~\LMC if and only if the instance~$I'$ is a yes-instance of~\LGGC. 
Since we only added degree-one vertices, the graph is still bipartite and~$2$-degenerate.
\fi
\begin{corollary}
For every~$c\geq 2$, \LGGC is~\W1-hard when parameterized by~$k$ on bipartite~$2$-degenerate graphs with unit weights.
\end{corollary}

\subsection{Hardness for the permissive version of~\LGGC}
Next, we present a running-time lower bound for~\LGGC based on the ETH.
This lower-bound holds even for the permissive version of~\LGGC.

\begin{lemma}\label{permissiveMaxCut}
Even the permissive version of~\LMC does not admit an~\FPT-algorithm when parameterized by~$k$, unless~$\FPT = \W{1}$ and cannot be solved in $f(k) \cdot n^{o(k)}$~time for any computable function~$f$, unless the ETH fails.
More precisely, this hardness holds even if there is an optimal solution in the~$k$-flip neighborhood of the initial solution.
\end{lemma}
\begin{proof}
We reduce from~\IS, where we are given an undirected graph~$G$ and~$k\in\mathds{N}$ and ask whether~$G$ contains an independent set of size~$k$. 
\IS is~\W1-hard when parameterized by the size~$k$ of the sought independent set even if the size of a largest independent set in the input graph is at most~$k$~\cite{DF13,C+15}.
Furthermore, even under these restrictions, \IS cannot be solves in $f(k) \cdot n^{o(k)}$~time for any computable function~$f$, unless the ETH fails~\cite{C+15}.

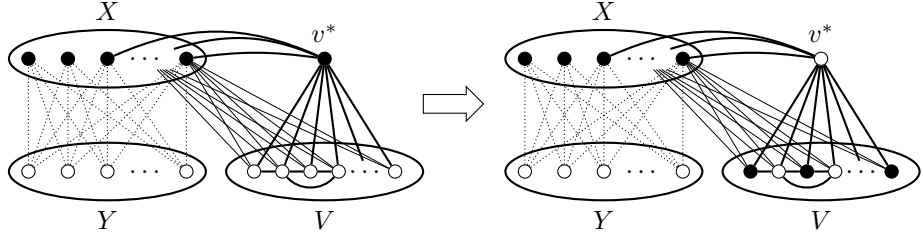
\begin{figure}

\begin{tikzpicture}[scale = .75, every fit/.style={ellipse,draw,inner sep=5pt,text width=1.5cm}
]

\tikzstyle{knoten}=[circle,fill=white,draw=black,minimum size=5pt,inner sep=0pt]
\tikzstyle{bl}=[fill=black]
\tikzstyle{localthick}=[thick]

\node[knoten,bl] (f1) at (0,0) {};
\node[knoten,bl] (f2) at ($(f1) + (.7,0)$) {};
\node[knoten,bl] (f3) at ($(f2) + (.7,0)$) {};
\node[] (f4) at ($(f3) + (.7,0)$) {$\dots$};
\node[knoten,bl] (f5) at ($(f4) + (.7,0)$) {};

\begin{scope}[yshift=-2cm]  
\node[knoten] (s1) at (0,0) {};
\node[knoten] (s2) at ($(s1) + (.7,0)$) {};
\node[knoten] (s3) at ($(s2) + (.7,0)$) {};
\node[] (s4) at ($(s3) + (.7,0)$) {$\dots$};
\node[knoten] (s5) at ($(s4) + (.7,0)$) {};
\end{scope}

\node [black,thick,fit=(f1) (f5),label=above:$X$] {};
\node [black,thick,fit=(s1) (s5),label=below:$Y$] {};


\foreach \i in {1,2,3,5}
    \foreach \j in {1,2,3,5}  
        \draw[densely dotted] (s\i) -- (f\j);

\begin{scope}[xshift=4cm]

\begin{scope}[yshift=-2cm]  
\node[knoten] (v1) at (0,0) {};
\node[knoten] (v2) at ($(v1) + (.5,0)$) {};
\node[knoten] (v3) at ($(v2) + (.5,0)$) {};
\node[knoten] (v4) at ($(v3) + (.5,0)$) {};
\node[]       (v5) at ($(v4) + (.5,0)$) {$\dots$};
\node[knoten] (v6) at ($(v5) + (.5,0)$) {};
\end{scope}

\node[knoten,bl,label=above:$v^*$] (vstar) at (1.25,0) {};

\foreach \j in {1,2,3,4,5,6}  
    \draw[localthick] (vstar) -- (v\j);

\foreach \j in {1,2,3,4,6}{  
    \draw[ultra thin] (f4) -- (v\j);
    \draw[ultra thin] (f5) -- (v\j);    
}

    \draw[localthick] (v1) -- (v2);    
    \draw[localthick] (v2) -- (v3);    
    \draw[localthick] (v3) -- (v4);    
    \draw[localthick] (v2) to [bend right=50] (v4);

\draw[localthick] (vstar) to [bend right=8] (f5);    
\draw[localthick] (vstar) to [bend right=20] (f4);    
\draw[localthick] (vstar) to [bend right=24] (f3);    
\end{scope}

\node [black,thick,fit=(v1) (v6),label=below:$V$] {};

    \node[single arrow, draw=black,
      minimum width = 8pt, single arrow head extend=3pt,
      minimum height=8mm] at ($(vstar) + (2.2,-.8)$) {};

\begin{scope}[xshift=8.8cm]

\node[knoten,bl] (f1) at (0,0) {};
\node[knoten,bl] (f2) at ($(f1) + (.7,0)$) {};
\node[knoten,bl] (f3) at ($(f2) + (.7,0)$) {};
\node[] (f4) at ($(f3) + (.7,0)$) {$\dots$};
\node[knoten,bl] (f5) at ($(f4) + (.7,0)$) {};

\begin{scope}[yshift=-2cm]  
\node[knoten] (s1) at (0,0) {};
\node[knoten] (s2) at ($(s1) + (.7,0)$) {};
\node[knoten] (s3) at ($(s2) + (.7,0)$) {};
\node[] (s4) at ($(s3) + (.7,0)$) {$\dots$};
\node[knoten] (s5) at ($(s4) + (.7,0)$) {};
\end{scope}

\node [black,thick,fit=(f1) (f5),label=above:$X$] {};
\node [black,thick,fit=(s1) (s5),label=below:$Y$] {};


\foreach \i in {1,2,3,5}
    \foreach \j in {1,2,3,5}  
        \draw[densely dotted] (s\i) -- (f\j);

\begin{scope}[xshift=4cm]

\begin{scope}[yshift=-2cm]  
\node[knoten,bl] (v1) at (0,0) {};
\node[knoten] (v2) at ($(v1) + (.5,0)$) {};
\node[knoten,bl] (v3) at ($(v2) + (.5,0)$) {};
\node[knoten] (v4) at ($(v3) + (.5,0)$) {};
\node[]       (v5) at ($(v4) + (.5,0)$) {$\dots$};
\node[knoten,bl] (v6) at ($(v5) + (.5,0)$) {};
\end{scope}

\node[knoten,label=above:$v^*$] (vstar) at (1.25,0) {};

\foreach \j in {1,2,3,4,5,6}  
    \draw[localthick] (vstar) -- (v\j);

\foreach \j in {1,2,3,4,6}{  
    \draw[ultra thin] (f4) -- (v\j);
    \draw[ultra thin] (f5) -- (v\j);    
}

    \draw[localthick] (v1) -- (v2);    
    \draw[localthick] (v2) -- (v3);    
    \draw[localthick] (v3) -- (v4);    
    \draw[localthick] (v2) to [bend right=50] (v4);

\draw[localthick] (vstar) to [bend right=8] (f5);    
\draw[localthick] (vstar) to [bend right=20] (f4);    
\draw[localthick] (vstar) to [bend right=24] (f3);    
\end{scope}

\node [black,thick,fit=(v1) (v6),label=below:$V$] {};

\end{scope}
\end{tikzpicture}
\caption{Two solution for the instance of~\LMC constructed in the proof of~\Cref{permissiveMaxCut}. 
In both solutions, the parts of the respective partitions are indicated by the color of the vertices.
The left partition shows the initial solution and the right partition shows an improving solution, if one exists.
The flip between these partitions is an independent set of size~$k$ in~$G$ together with the vertex~$v^*$.}
\label{figure max cut permissive}
\end{figure}

Let~$I:=(G=(V,E),k)$ be an instance of~\IS and let~$n:= |V|$ and~$m:= |E|$.
We construct an equivalent instance~$I':=(G'=(V',E'),\omega',A,B,k')$ of~\LMC with $\omega': E' \to \{1\}$ as follows.
We initialize~$G'$ as a copy of~$G$.
Next, we add two vertex sets~$X$ and~$Y$ of size~$n^3$ each to~$G'$.
Additionally, we add a vertex~$v^*$ to~$G'$. 
Next, we describe the edges incident with at least one newly added vertex.
We add edges to~$G'$ such that~$v^*$ is adjacent to each vertex of~$V$ and~$n-k+1$ arbitrary vertices of~$X$.
Moreover, we add edges to~$G'$ such that each vertex of~$X$ is adjacent to each vertex of~$Y$.
Finally, we add edges to~$G'$ such that each vertex~$v\in V$ is adjacent with exactly~$|N_G(v)|$ arbitrary vertices of~$X$ in~$G'$.
This completes the construction of~$G'$.
It remains to define the initial partition~$(A,B)$ of~$V'$ and the search radius~$k'$.
We set~$A := V \cup Y$, $B := X \cup \{v^*\}$, and~$k' := k + 1$.
This completes the construction of~$I'$.
Note that each vertex of~$V$ has exactly one neighbor more in~$B$ than in~$A$.
Moreover, $v^*$ has~$2k-1$ more neighbors in~$A$ than in~$B$.

Intuitively, the only way to improve over the partition~$(A,B)$ is to flip an independent set in~$G$ of size~$k$ from~$A$ to~$B$.
Then, $v^*$ has exactly one neighbor more in~$B$ than in~$A$ and flipping~$v^*$ from~$B$ to~$A$ improves over the initial partition~$(A,B)$ by exactly one edge.
See Figure~\ref{figure max cut permissive} for an illustration.
Next, we show that this reduction is correct.

$(\Rightarrow)$
Let~$S$ be an independent set of size~$k$ in~$G$.
We set~$A':= (A \setminus S) \cup \{v^*\}$ and~$B':= V'\setminus A' = (B \cup S) \setminus \{v^*\}$ and show that~$(A',B')$ improves over~$(A,B)$.
Note that~$E(A,B) = E(Y,X) \cup E(V,\{v^*\}) \cup E(V,X)$ which implies that~$$|E(A,B)| = |Y| \cdot |X| + |V| + \sum_{v\in V} |N_G(v)| = n^6 + n + 2m.$$
Moreover, note that
$$E(A',B') = E(Y,X) \cup E(\{v^*\},S) \cup E(\{v^*\},X) \cup E(V\setminus S,S) \cup E(V\setminus S,X).$$
Since~$S$ is an independent set in~$G$ and~$G'$, for each vertex~$v\in S$, $V\setminus S$ contains all neighbors of~$v$ in~$G$.
Consequently, $|E(V\setminus S,S)| = \sum_{v\in S} |N_G(v)|$.
Since~$S$ has size~$k$, the above implies that 
\begin{align*}
|E(A',B')|&= |Y| \cdot |X| + |S| + n-k+1 + \sum_{v\in S} |N_G(v)| + \sum_{v\in V\setminus S} |N_G(v)|\\
&= n^6 + n + 2m + 1.
\end{align*}
Hence, the partition~$(A',B')$ improves over the partition~$(A,B)$ which implies that~$I'$ is a yes-instance of~\LMC.

$(\Leftarrow)$
Let~$(A',B')$ be an optimal partition of~$G'$ and suppose that~$(A',B')$ improves over~$(A,B)$.
We show that~$I$ is a yes-instance of~\IS.
Due to the first direction, this then implies that~$I'$ is a yes-instance of~\LMC which further implies that the initial partition~$(A,B)$ is an optimal partition for~$G'$ if and only if~$(A,B)$ is~$k'$-flip optimal.
To show that~$I$ is a yes-instance of~\IS, we first analyze the structure of the optimal partition~$(A',B')$ for~$G'$.

First, we show that all vertices of~$X$ are on the opposite part of the partition~$(A',B')$ than all vertices of~$Y$.
\begin{claim}\label{claim max cut permissive}
$A'$ contains all vertices of~$Y$ and~$B'$ contains all vertices of~$X$, or $A'$ contains all vertices of~$X$ and~$B'$ contains all vertices of~$Y$.
\end{claim}
\begin{claimproof}
We show that if neither of these statements holds, then~$E(A',B')$ contains less edges than~$E(A,B)$.
This would then contradict the fact that~$(A',B')$ is an optimal partition.
We distinguish two cases.

If all vertices of~$X\cup Y$ are on the same part of the partition~$(A',B')$, then~$E(A',B')$ contains at most~$|N_{G'}(V)| + |N_{G'}(v^*)| < n^3$ edges.
Hence, $E(A',B')$ contains strictly less edges than~$E(A,B)$ which contradicts the fact that~$(A',B')$ is an optimal partition.
Otherwise, if not all vertices of~$X\cup Y$ are on the same part of the partition~$(A',B')$ and not all vertices of~$X$ are on the opposite part of the partition than all vertices of~$Y$, then both~$A'$ and~$B'$ contain at least one vertex of~$X$ each, or both~$A'$ and~$B'$ contain at least one vertex of~$Y$ each.
In both cases, at least~$\min(|X|,|Y|) = n^3$ edges of~$E(X,Y)$ are not contained in~$E(A',B')$.
Hence, $E(A',B')$ has size at most~$|E| - n^3$.
Again, since~$E$ contains at most~$|N_{G'}(V)| + |N_{G'}(v^*)| < n^3$ edges outside of~$E(X,Y)$, this implies that~$E(A',B')$ contains strictly less edges than~$E(A,B)$.
This contradicts the fact that~$(A',B')$ is an optimal partition.
Consequently, the statement holds.
\end{claimproof}\\

By Claim~\ref{claim max cut permissive}, we may assume without loss of generality that~$A'$ contains all vertices of~$Y$ and~$B'$ contains all vertices of~$X$.
Next, we show that~$v^*$ is contained in~$A'$.
Assume towards a contradiction that~$v^*$ is contained in~$B'$.
Hence, each vertex~$v\in V$ has~$|N_G(v)|+1$ neighbors in~$B'$.
By construction, each vertex~$v\in V$ has exactly~$2\cdot |N_G(v)|+1$ neighbors in~$G'$.
Hence, each vertex of~$V$ has more neighbors in~$B'$ than in~$A'$.
Consequently, $A'$ contains all vertices of~$V$, since~$(A',B')$ is an optimal partition for~$G'$.
This implies that~$A' := Y \cup V = A$ and~$B' := X \cup \{v^*\} = B$ which contradicts the assumption that~$(A',B')$ improves over~$(A,B)$.
Consequently, $v^*$ is contained in~$A'$ together with all vertices of~$Y$.
It remains to determine the partition of the vertices of~$V$ into~$A'$ and~$B'$.
 
Let~$S := B' \cap V$.
We show that~$S$ is an independent set of size~$k$ in~$G$.
First, assume towards a contradiction that~$S$ is not an independent set in~$G$.
Then, there are two adjacent vertices~$u$ and~$w$ of~$V$ in~$B'$.
Hence, $u$ has at least~$|N_G(u)|+1$ neighbors in~$B'$, since~$u$ is adjacent to~$|N_G(u)|$ vertices of~$X$.
Since the degree of~$u$ in~$G'$ is~$2\cdot |N_G(u)|+1$, flipping vertex~$u$ from~$B'$ to~$A'$ would result in an improving solution.
This contradicts the fact that~$(A',B')$ is an optimal partition of~$G'$.
Hence, $S$ is an independent set in~$G$.
By assumption, the size of the largest independent set in~$G$ is at most~$k$.
Hence, to show that~$S$ is an independent set of size~$k$, it suffices to show that~$S$ has size at least~$k$.
To this end, we analyze the number of edges of~$E(A',B')$.
Recall that~$A' := (A\setminus S) \cup \{v^*\}$ and~$B' := (B\cup S)\setminus \{v^*\}$.
Hence, analogously to the first direction of the correctness proof, $E(A',B') = E(Y,X) \cup E(\{v^*\},S) \cup E(\{v^*\},X) \cup E(V\setminus S,S) \cup E(V\setminus S,X)$.
Since~$S$ is an independent set, this implies that 
\begin{align*}
|E(A',B')| & = |Y| \cdot |X| + |S| + n-k+1 + \sum_{v\in S} |N_G(v)| + \sum_{v\in V\setminus S} |N_G(v)|\\
&= n^6 + n -k + 1 + 2m + |S|.
\end{align*}
Since we assumed that the partition~$(A',B')$ improves over~$(A,B)$ and~$|E(A,B)| = n^6 + n + 2m$, this implies that~$S$ has size at least~$k$.
Consequently, $S$ is an independent set of size~$k$ in~$G$, which implies that~$I$ is a yes-instance of~\IS.

This also implies that, if~$(A,B)$ is not an optimal partition for~$G'$, then there is an optimal partition for~$G'$ with flip-distance exactly~$k'$ from~$(A,B)$.
Hence, the reduction is correct. 
Recall that~\IS is~\W1-hard when parameterized by~$k$ and cannot be solved in $f(k)\cdot n^{o(k)}$~time for any computable function~$f$, unless the ETH fails.
Since~$|V'| \in \Oh(n^3)$ and~$k' \in \Oh(k)$, this implies that the permissive version of~\LMC (i)~does not admit an~\FPT-algorithm when parameterized by~$k'$, unless~$\FPT = \W{1}$ and (ii)~cannot be solved in $f(k')\cdot |V'|^{o(k')}$~time for any computable function~$f$, unless the ETH fails.
\end{proof}

These intractability results can be transferred to each larger constant value of~$c$.

\begin{theorem}\label{permissiveMaxCCut}
For every~$c \geq 2$, even the permissive version of~\LGGC does not admit an~\FPT-algorithm when parameterized by~$k$, unless~$\FPT = \W{1}$ and cannot be solved in $f(k) \cdot n^{o(k)}$~time for any computable function~$f$, unless the ETH fails.
This holds even on graphs with unit weights.
\end{theorem}

\begin{proof}
Let~$I:=(G=(V,E),\omega,A,B,k)$ be an instance of~\LMC with~$\omega(e) = 1$ for each edge~$e\in E$, such that there is an optimal partition~$(A',B')$ for~$G$ with flip-distance at most~$k$ with~$(A,B)$ and let~$n:= |V|$ and~$m:= |E|$.
Due to~\Cref{permissiveMaxCut}, even under these restrictions, \LMC does not admit an~\FPT-algorithm when parameterized by~$k$, unless $\FPT = \W1$ and cannot be solved in $f(k) \cdot n^{o(k)}$~time for any computable function~$f$, unless the ETH fails.

Fix a constant~$c > 2$.
We describe how to obtain in polynomial time an equivalent instance~$I':=(G':=(V',E'),c,\omega', \colo,k)$ of~\LGGC.
We obtain the graph~$G'$ by adding for each~$i\in [1,c]$ an independent set~$X_i$ of size~$n^2$ to~$G$ and adding edges such that each vertex of~$X_i$ is adjacent with each vertex of~$\{v\in X_j\mid j\in [1,c]\setminus \{i\}\}$.
Additionally, for each~$i\in [3,c]$, we add edges between each vertex of~$X_i$ and each vertex of~$V$.
This completes the construction of~$G'$.
Again, each edge receives weight 1 with respect to the weight function~$\omega'$.
Finally, we define the~$c$-coloring~$\colo$ of~$G'$.
For each~$i\in[1,c]$, we set~$\colo(v) := i$ for each vertex~$v\in X_i$.
Additionally, for each vertex~$v\in A$, we set~$\colo(v) := 1$ and for each vertex~$w\in B$, we set~$\colo(w) := 2$.
This completes the construction of~$I'$.

Note that~$E'(\colo)$ contains all edges of~$E'\setminus E$ and thus misses less than~$n^2$ edges of~$G'$ in total.
Intuitively, this ensures that only vertices of~$V$ may flip their color and only to the colors~1 or~2, since in each other~$c$-coloring, at least~$n^2$ edges are missing.
In other words, the large independent sets~$X_i$ ensure that to improve over~$\colo$, one can only flip vertices of~$V$ from color~1 to~2 or vice versa.
This is then improving if and only if the corresponding flip on the~\LMC-instance~$I$ is improving.

Next, we show the correctness of the reduction.

$(\Rightarrow)$
Let~$(A',B')$ be an optimal partition of~$G$ that improves over~$(A,B)$ and let~$S$ be the flip between~$(A,B)$ and~$(A',B')$.
By assumption, we know that~$S$ has size at most~$k$.
We define a~$c$-coloring~$\colo^*$ for~$G'$ as follows:
The colorings~$\colo$ and~$\colo^*$ agree on all vertices of~$V'\setminus S$, for each vertex~$v\in A\cap S$, we set~$\colo^*(v) := 2$, and for each vertex~$w\in B\cap S$, we set~$\colo^*(w) := 1$.
Note that~$\Dflip(\colo,\colo^*) = S$ which implies that~$\colo^*$ and~$\colo$ have flip-distance at most~$k$.
Moreover, note that~$A' = \{v\in V\mid \colo^*(v) = 1\}$ and~$B' = \{v\in V\mid \colo^*(v) = 2\}$.
It remains to show that~$\colo^*$ improves over~$\colo$.
To this end, note that both~$E'(\colo)$ and~$E'(\colo^*)$ contain all edges of~$E'\setminus E$.
Hence, $\colo^*$ improves over~$\colo$ if and only if~$|E'(\colo^*)\cap E| > |E'(\colo)\cap E|$.
Note that~$E'(\colo^*)\cap E = E(A',B')$ and~$E'(\colo)\cap E = E(A,B)$.
Hence, the assumption that~$(A',B')$ is a better partition for~$G$ than~$(A,B)$ implies that~$\colo^*$ improves over~$\colo$.
Consequently, $I'$ is a yes-instance of~\LGGC.

$(\Leftarrow)$
Let~$\colo^*$ be an optimal~$c$-coloring for~$G'$ and assume that~$\colo^*$ improves over~$\colo$.
To show that there is a better partition for~$G$ than~$(A,B)$, we first prove that each optimal~$c$-coloring~$\colo^*$ for~$G'$ contains all edges of~$E'\setminus E$.
\begin{claim}
It holds that $E'(\colo^*)$ contains all edges of~$E'\setminus E$. 
\end{claim}
\begin{claimproof}
Assume towards a contradiction that~$E'(\colo^*)$ does not contain all edges of~$E'\setminus E$.
Since~$E'(\colo^*)$ does not contain all edges of~$E'\setminus E$, there is an~$i\in [1,c]$ and a vertex~$x\in X_i$, such that at least one edge incident with~$x$ is not contained in~$E'(\colo^*)$.
We distinguish two cases.

\textbf{Case 1:} There is a vertex~$y\in X_i$, such that each edge incident with~$y$ is contained in~$E'(\colo^*)$\textbf{.}
Consider the~$c$-coloring~$\colo'$ for~$G'$ obtained from~$\colo^*$ by flipping the color of vertex~$x$ to~$\colo^*(y)$.
Since~$x$ and~$y$ have the exact same neighborhood by definition of~$G'$ and are not adjacent, each edge incident with~$x$ is contained in~$E'(\colo')$.
Consequently, $\colo'$ is a better~$c$-coloring for~$G'$ than~$\colo^*$.
This contradicts the fact that~$\colo^*$ is an optimal~$c$-coloring for~$G'$.

\textbf{Case 2:} Each vertex of~$X_i$ is incident with at least one edge that is not contained in~$E'(\colo^*)$\textbf{.} 
Since~$X_i$ is an independent set in~$G'$, this directly implies that~$E'(\colo^*)$ misses at least~$|X_i| = n^2$ edges of~$E'$.
Hence, $E(\colo^*)$ contains less edges than~$E'(\colo)$ and is thus not an optimal~$c$-coloring for~$G'$, a contradiction.
\end{claimproof}\\

By the above, we know that~$\colo^*$ contains all edges of~$E'\setminus E$.
Since~$G' - V$ is a complete~$c$-partite graph with~$c$-partition~$(X_1, \dots, X_c)$, there is a bijection~$\pi\colon [1,c]\to [1,c]$, such that for each~$i\in [1,c]$, each vertex of~$X_i$ receives color~$\pi(i)$ under~$\colo^*$.
Moreover, since for each~$j\in [3,c]$, each vertex~$v\in V$ is adjacent with each vertex of~$X_j$, each vertex of~$V$ receives either color~$\pi(1)$ or color~$\pi(2)$ under~$\colo^*$.
For simplicity, we assume in the following that~$\pi$ is the identity function, that is, for each~$i\in [1,c]$, each vertex of~$X_i$ receives color~$i$ under~$\colo^*$ and each vertex of~$V$ receives either color~1 or color~2 under~$\colo^*$.
Let~$A' := \{v\in V \mid  \colo^*(v) = 1\}$ and let~$B':=\{v\in V \mid  \colo^*(v) = 2\}$.
Note that~$|E'(\colo^*)| = |E'\setminus E| + |E(A',B')|$ and that~$|E'(\colo)| = |E'\setminus E| + |E(A,B)|$.
Hence, $\colo^*$ is a better~$c$-coloring for~$G'$ than~$\colo$ if and only if~$(A', B')$ is a better partition of~$G$ than~$(A,B)$.
Since~$\colo^*$ improves over~$\colo$, this implies that~$(A,B)$ is not an optimal partition for~$G$. 
By assumption there is an optimal partition for~$G$ having flip-distance at most~$k$ with~$(A,B)$.
Hence, $I$ is a yes-instance of~\LMC.
By the first direction, this further implies that~$I'$ is a yes-instance of~\LGGC if~$\colo$ is not an optimal~$c$-coloring for~$G'$.

Note that this implies that there is an optimal~$c$-coloring for~$G'$ having flip-distance at most~$k$ with~$\colo$.
Hence, the reduction is also correct for the permissive version of~\LGGC.
Recall that~\LMC does not admit an~\FPT-algorithm when parameterized by~$k$, unless $\FPT = \W1$ and cannot be solved in $f(k) \cdot n^{o(k)}$~time for any computable function~$f$, unless the ETH fails.
Since~$|V'|\in n^{\Oh(1)}$, this implies that even the permissive version of~\LGGC does not admit an~\FPT-algorithm when parameterized by~$k$, unless $\FPT = \W1$ and cannot be solved in $f(k) \cdot |V'|^{o(k)}$~time for any computable function~$f$, unless the ETH fails.
\end{proof}

\subsection{Hardness for related partition problems}
Based on~\Cref{permissiveMaxCut}, we are also able to show hardness for a previously considered local search version of~\textsc{Min Bisection} and~\textsc{Max Bisection}~\cite{FFL+12}.
In both these problems, the input is again an undirected graph~$G$ and the goal is to find a balanced partition~$(X,Y)$ of the vertex set of~$G$ that minimizes (maximizes) the edges in~$E(X,Y)$.
Here, a partition~$(X,Y)$ is~\emph{balanced} if the size of~$X$ and the size of~$Y$ differ by at most one.
Due to the close relation to~\LMC, the proposed local neighborhood for these problems is also the~$k$-flip-neighborhood.
The corresponding local search problems in which we ask for a better balanced partition of flip-distance at most~$k$ are denoted by~\textsc{LS Min Bisection} and~\textsc{LS Max Bisection}, respectively.
It was shown that both these local search problems can be solved in $2^{\Oh(k)}\cdot n^{\Oh(1)}$~time on restricted graph classes~\cite{FFL+12} but~\W1-hardness on general graphs was not shown so far.
 
\begin{corollary}
\textsc{LS Min Bisection} and~\textsc{LS Max Bisection} are~\W1-hard when parameterized by~$k$ and cannot be solved in $f(k) \cdot n^{o(k)}$~time for any computable function~$f$, unless the~ETH fails.
This running time lower-bound holds even for the permissive version of both problems and both permissive versions do not admit~\FPT-algorithms when parameterized by~$k$, unless $\FPT = \W1$.
\end{corollary} 
\begin{proof}
First, we show the statement for~\textsc{LS Max Bisection}.
Afterwards, we discuss how to obtain the similar intractability result for~\textsc{LS Min Bisection}.
Let~$I:=(G=(V,E),\omega,A,B,k)$ be an instance of~\LMC with~$\omega(e) = 1$ for each edge~$e\in E$, such that there is an optimal partition~$(A',B')$ for~$G$ with flip-distance at most~$k$ with~$(A,B)$ and let~$n:= |V|$.
Due to~\Cref{permissiveMaxCut}, even under these restrictions~\LMC is~\W1-hard when parameterized by~$k$ and cannot be solved in $f(k) \cdot n^{o(k)}$~time for any computable function~$f$, unless the ETH fails.

We obtain an equivalent instance~$I':=(G',X,Y,k')$ of~\textsc{LS Max Bisection}, by simply adding a large independent set to~$G$.
That is, we obtain the graph~$G'= (V',E')$ by adding a set~$Z$ of~$n+2k$ isolated vertices to~$G$, setting~$X:= A \cup Z_A$ and~$Y := B \cup (Z \setminus Z_A)$ for some arbitrary vertex set~$Z_A\subseteq Z$ of size~$|B|+k$, and setting~$k':= 2k$.
Note that~$G$ and~$G'$ have the exact same edge set and that~$X$ and~$Y$ have the same size and contain at least~$k$ vertices of~$Z$ each.
Intuitively, this ensures that we can perform an improving~$k$-flip on the vertices of~$V$ and afterwards end back at a balanced partition by flipping at most~$k$ additional vertices of~$Z$ to the smaller part of the resulting potentially not balanced partition.

Note that for each partition~$(X',Y')$ of~$G'$, $E'(X',Y') = E(X\cap V,Y\cap V)$.
This directly implies that~$(X,Y)$ is an optimal balanced partition for~$G'$ if~$(X\cap V,Y\cap V) = (A,B)$ is an optimal partition for~$G$.
Hence, $I'$ is a no-instance of~\textsc{LS Max Bisection} if~$I$ is a no-instance of~\LMC, since by assumption, $(A,B)$ is an optimal partition for~$G$ if and only if~$I$ is a no-instance of~\LMC.
It thus remains to show that~$I'$ is a yes-instance of~\textsc{LS Max Bisection} if~$I$ is a yes-instance of~\LMC.
By the above, this then implies that~$I'$ is a no-instance of~\textsc{LS Max Bisection} if and only if~$(X,Y)$ is an optimal balanced partition for~$G'$.

Assume that~$I$ is a yes-instance of~\LMC.
This implies that~$(A,B)$ is not an optimal partition of~$G$.
Let~$(A',B')$ be an optimal partition of~$G$.
By assumption, we can assume that~$(A',B')$ has flip-distance at most~$k$ with~$(A,B)$.
Let~$\widehat{X} := A' \cup Z_A$ and~$\widehat{Y} := B' \cup (Z \setminus Z_A)$.
Note that~$(\widehat{X}, \widehat{Y})$ has flip-distance at most~$k$ with~$(X,Y)$ and is a better partition for~$G'$ than~$(X,Y)$.
Still, $(\widehat{X}, \widehat{Y})$ might not be a balanced partition.
But since~$(\widehat{X}, \widehat{Y})$ has flip-distance at most~$k$ with~$(X,Y)$, the size of~$\widehat{X}$ and the size of~$\widehat{Y}$ differ by at most~$k$.
Hence, we can obtain a balances partition~$(X',Y')$ for~$G'$ by flipping at most~$k$ vertices from~$Z$ from the larger part of the partition to the smaller part.
This is possible, since by construction both~$X$ and~$Y$ contain at least~$k$ vertices of~$Z$.
Since~$(X',Y')$ is obtained from~$(\widehat{X}, \widehat{Y})$ by flipping only isolated vertices, $(X',Y')$ is also a better partition for~$G'$ than~$(X,Y)$ and has flip-distance at most~$2k = k'$ with~$(X,Y)$.
Consequently, $I'$ is a yes-instance of~\textsc{LS Max Bisection} if~$I$ is a yes-instance of~\LMC.

Hence, $I'$ is a no-instance of~\textsc{LS Max Bisection} if and only if~$(X,Y)$ is an optimal balanced partition for~$G'$.
This implies that the reduction also works for the permissive version of~\textsc{LS Max Bisection}.
Recall that~\LMC is~\W1-hard when parameterized by~$k$ and cannot be solved in $f(k) \cdot n^{o(k)}$~time for any computable function~$f$, unless the ETH fails.
Since~$|V'|\in n^{\Oh(1)}$ and~$k' \in \Oh(k)$, this implies that (i)~the strict version of~\textsc{LS Max Bisection} is~\W1-hard when parameterized by~$k'$, (ii)~the permissive version of~\textsc{LS Max Bisection} does not admit an~\FPT-algorithm when parameterized by~$k'$, unless $\FPT = \W1$, and (iii)~both versions of~\textsc{LS Max Bisection} cannot be solved in $f(k) \cdot |V'|^{o(k')}$~time for any computable function~$f$, unless the ETH fails.

The reduction to~\textsc{LS Min Bisection} works analogously by not considering the graph~$G'$ as the input graph, but the complement graph~$G'':= (V',E'')$ of~$G'$.
Consequently, for each balanced partitions~$(X',Y')$ of~$G''$, $|E''(X',Y')| = |V'|^2/4- |E'(X',Y')|$.
In other words, $(X',Y')$ is a better partition for~$G''$ than~$(X,Y)$ if and only if~$(X',Y')$ is a better partition for~$G'$ than~$(X,Y)$.
Hence, the intractability results also hold for the strict and permissive versions of~\textsc{LS Min Bisection}.
\end{proof}

Additionally, based on the close relation of~\textsc{Max 2-SAT} and~\MC, we can also transfer new hardness results to the strict and permissive version of local search for~\textsc{Max SAT} with respect to the~$k$-flip-neighborhood.
This problem was considered by Szeider~\cite{Szei11} under the name of~\textsc{$k$-Flip Max Sat}.
Here, the input is a boolean formula~$F$ in CNF, an assignment~$\tau$ of the variables of~$F$, and an integer~$k$, and we ask for an assignment~$\tau'$ of the variables of~$F$ for which~$\dflip(\tau,\tau')\leq k$ and that satisfies more clauses of~$F$ than~$\tau$.
Szeider~\cite{Szei11} showed that (i)~the strict version of~\textsc{$k$-Flip Max Sat} is~\W1-hard when parameterized by~$k$ even on formulas where each clause has size two and (ii)~the permissive version of~\textsc{$k$-Flip Max Sat} does not admit an \FPT-algorithm when parameterized by~$k$, unless~$\FPT = \W1$, even when each clause has size at most three and the formula is either Horn or anti-Horn\footnote{Here, a formula is Horn (anti-Horn), if each clause contains at most one positive (negative) literal.}.

\begin{theorem}\label{flip sat implications}
Even on formulas~$F$ where each clause has size two and contains exactly one positive and one negative literal, both the strict and the permissive versions of~\textsc{$k$-Flip Max Sat} cannot be solved in $f(k) \cdot |F|^{o(k)}$~time for any computable function~$f$, unless the ETH fails, the strict version of~\textsc{$k$-Flip Max Sat} is~\W1-hard when parameterized by~$k$, and the permissive version of~\textsc{$k$-Flip Max Sat}  does not admit an~\FPT-algorithm when parameterized by~$k$, unless $\FPT = \W1$.
\end{theorem}
\begin{proof}
We present a reduction from~\LGGC[2], for which the desired intractability results hold even for the permissive version due to~\Cref{permissiveMaxCut}.
Let~$I:=(G=(V,E),\omega,\colo,k)$ be an instance of~\LMC with~$\omega(e) = 1$ for each edge~$e\in E$.
We define a formula~$F$ as follows:
The variables of~$F$ are exactly the vertices of~$V$ and for each edge~$\{u,v\}\in E$, $F$ contains the clauses~$\{u, \neg v\}$ and~$\{\neg u, v\}$.

Let~$\tau:V \to \{1,2\}$ be a~$2$-coloring of~$V$.
We interpret~$\tau$ as an assignment for~$F$, where color~$1$ ($2$) represents the truth value ``true'' (``false''). 
Note that by construction, for each edge~$\{u,v\}\in E$, $\tau$ satisfies at least one of the clauses~$\{u, \neg v\}$ and~$\{\neg u, v\}$.
Moreover, $\tau$ satisfies both clauses~$\{u, \neg v\}$ and~$\{\neg u, v\}$ if and only if the edge~$\{u,v\}$ is properly colored under~$\tau$.
This implies that~$\tau$ satisfies~$|E| + |E(\tau)|$ clauses of~$F$.
Consequently, an assignment~$\tau'$ of~$F$ satisfies more clauses of~$F$ that the~$2$-coloring~$\colo$ if and only if~$|E(\tau')| > |E(\colo)|$.
Hence, $I$ is a yes-instance of~\LGGC[2] if and only if~$(F,\colo,k)$ is a yes-instance of~\textsc{$k$-Flip Max Sat}.
\end{proof}

Hence, in comparison to the hardness results presented by Szeider~\cite{Szei11}, \Cref{flip sat implications} provides a tight ETH lower bound as well as hardness for formulas that are 2-Sat, Horn and anti-Horn simultaneously.

\section{Algorithms}
In this section, we complement the running time lower bound of~\Cref{permissiveMaxCCut} by presenting an algorithm for~\LGGC that runs in $\Delta^{\Oh(k)} \cdot c \cdot n$~time, where~$\Delta$ denotes the maximum degree of the input-graph. 
Our algorithm for \LGGClong{} follows a simple framework: Generate a collection of candidate sets~$S$ that may improve the coloring if the vertices in~$S$ flip their colors. 
For each such candidate set~$S$, we only know that the colors of the vertices of~$S$ change, but we do not yet know which new color the vertices receive. 
To answer this question, that is, to find whether there is any coloring of~$S$ that leads to an improving coloring, we present an algorithm based on dynamic programming.

We first describe the subroutine that we use to find a best coloring for a given candidate set~$S$ of vertices to flip.

\begin{theorem}\label{dp}
Let~$G=(V,E)$ be a graph, let~$\omega:E\to \mathds{Q}$ be an edge-weight function, let~$\colo$ be a~$c$-coloring of~$G$, and let~$S\subseteq V$ be a set of size at most~$k$.
One can compute in $\Oh(3^k\cdot c\cdot k^2+k\cdot \Delta(G))$~time a~$c$-coloring~$\colo'$ of~$G$ such that~$\Dflip(\colo,\colo')\subseteq S$ and~$\omega(E(\colo'))$ is maximal among all such colorings.
\end{theorem}
\begin{proof}
We use  dynamic programming.
Initially, we compute for each vertex~$v\in S$ and each color~$i\in[1,c]$ the weight~$\theta_v^i$ of edges between~$v$ and vertices of~$V\setminus S$ that do not receive color~$i$ under~$\colo$, that is, $\theta_v^i := \omega(\{\{v,w\}\in E\mid w\in N(v)\setminus S, \colo(w)\neq i\})$.
Moreover, we compute the weight~$\omega_{S}$ of all properly colored edges of~$E(S,N[S])$ as~$\omega_{S} :=\omega(\{\{u,v\}\in E(S,N[S])\mid \colo(u) \neq \colo(v)\})$.
This can be done in $\Oh(c\cdot k + k\cdot \Delta(G))$~time.

The table~$T$ has entries of type~$T[S',c']$ for each vertex set~$S'\subseteq S$ and each color~$c' \in [1, c]$. 
Each entry~$T[S',c']$ stores the maximum total weights of properly colored edges with at least one endpoint in~$S'$ and no endpoint in~$S \setminus S'$ such that the following holds:
\begin{enumerate}\setlength{\itemsep}{0em}
\item[1.] the vertices in~$S'$ have some color in~$[1,c']$, and
\item[2.] every vertex~$v \in V \setminus S$ has color~$\colo(v)$. 
\end{enumerate}
 
 We start to fill the dynamic programming table by setting
 $T[S',1] := \sum_{v\in S'}\theta_v^1$  for each vertex set~$S'\subseteq S$.
 
 For each vertex set~$S'\subseteq S$ and each color~$c'\in[2,c]$, we set
\begin{align*}
T[S',c'] := &\max_{S''\subseteq S'} T[S'\setminus S'', c'-1] + \omega(E(S'', S'\setminus S'')) + \sum_{v\in S''}\theta_v^{c'}.
\end{align*} 
 
Intuitively, to find the best way to assign colors of~$[1,c']$ to the vertices of~$S'$, we search for the best vertex set~$S''\subseteq S'$, assign color~$c'$ to all vertices of~$S''$, and find the best way to assign the colors of~$[1,c'-1]$ to the vertices of~$S'\setminus S''$.
The maximal improvement~$\omega(E(\colo'))-\omega(E(\colo))$ for any~$c$-coloring~$\colo'$ with~$\Dflip(\colo,\colo')\subseteq S$ can then be found by evaluating~$T[S,c] - \omega_S$: this term corresponds to the maximum total weight of properly colored edges we get when distributing the vertices of~$S$ among all color classes minus the original weights when every vertex of~$S$ sticks with its color under~$\colo$.
The corresponding~$c$-coloring can be found via traceback.
 
The formal correctness proof is straightforward and thus omitted.
Hence, it remains to show the running time.
The dynamic programming table~$T$ has~$2^k \cdot c$ entries. 
Each of these entries can be computed in $\Oh(2^{|S'|} \cdot k^2)$~time. 
Consequently, all entries can be computed in $\Oh(\sum_{i= 0}^k \binom{k}{i}\cdot 2^i \cdot c\cdot k^2) = \Oh(3^k \cdot c\cdot k^2)$~time in total.
Hence, the total running time is $\Oh(3^k \cdot c\cdot k^2 + k\cdot \Delta(G))$.
\end{proof}

For~\LMC, if we enforce that each vertex of~$S$ changes its color, the situation is even simpler: 
When given a set~$S\subseteq V$ of~$k$ vertices that must flip their colors, the best possible improvement can be computed in~$\Oh(k\cdot\Delta(G))$~time, since every vertex of~$S$ must replace its color with the unique other color.


Recall that the idea of our algorithms for~\LMC and~\LGGC is to iterate over possible candidate sets of vertices that may flip their colors. 
With the next lemma we show that it suffices to consider those vertex sets that are connected in the input graph.

\begin{lemma}\label{lem connected}
Let~$I:=(G=(V,E), c,\omega, \colo,k)$ be an instance of~\LGGC. 
Then, for every inclusion-minimal improving~$k$-flip~$\colo'$ for~$\colo$, the vertex set~$\Dflip(\colo,\colo')$ is connected in~$G$.
\end{lemma}
\begin{proof}
Let~$\colo'$ be an inclusion-minimal improving~$k$-flip for~$\colo$.
Let~$S' := \Dflip(\colo,\colo')$ be the vertices~$\colo$ and~$\colo'$ do not agree on and let~$\mathcal{C}$ denote the connected components in~$G[S']$.
We show that if there are at least two connected components in~$\mathcal{C}$, then there is an improving~$k$-neighbor~$\widetilde{\colo}$ of~$\colo$ with~$\Dflip(\colo, \widetilde{\colo}) \subsetneq \Dflip(\colo,\colo')$.
For each connected component~$C\in\mathcal{C}$, let~$E^+_C := (E(C, V) \cap E(\colo')) \setminus E(\colo)$ denote the set of properly colored edges in~$E(\colo')\setminus E(\colo)$ that have at least one endpoint in~$C$ and let~$E^-_C := (E(C, V) \cap E(\colo)) \setminus E(\colo')$ denote the set of properly colored edges in~$E(\colo)\setminus E(\colo')$ that have at least one endpoint in~$C$.
Note that~$E(\colo')\setminus E(\colo) = \sum_{C \in\mathcal{C}} E^+_C$ and that~$E(\colo)\setminus E(\colo') = \sum_{C \in\mathcal{C}} E^-_C$.
Hence, the improvement of~$\colo'$ over~$\colo$ is $$\omega(E(\colo')) - \omega(E(\colo)) = \sum_{C \in\mathcal{C}} \omega(E^+_C) - \sum_{C \in\mathcal{C}} \omega(E^-_C) = \sum_{C \in\mathcal{C}} (\omega(E^+_C) - \omega(E^-_C))$$
Since~$\colo'$ improves over~$\colo$, this implies that there is at least one connected component~$S\in\mathcal{C}$ with~$\omega(E^+_S) - \omega(E^-_S) > 0$.
Let~$\widetilde{\colo}$ be the~$c$-coloring of~$G$ that agrees with~$\colo$ on all vertices of~$V\setminus S$ and agrees with~$\colo'$ on all vertices of~$S$.
Hence, $\widetilde{\colo}$ is an improving~$k$-neighbor of~$\colo$ with~$\Dflip(\colo,\widetilde{\colo}) \subsetneq \Dflip(\colo,\colo')$. 
\end{proof}


We next combine \Cref{dp} and \Cref{lem connected} to obtain our algorithm for~\LGGC.
\begin{theorem}
\LGGC can be solved in $\Oh((3\cdot e)^k \cdot (\Delta(G)-1)^{k+1} \cdot c \cdot k^3 \cdot n)$~time.
\end{theorem}
\begin{proof}
Let~$I=(G,c,\omega, \colo,k)$ be an instance of \LGGC.
By~\Cref{lem connected},~$I$ is a yes-instance of~\LGGC if and only if~$\colo$ has an improving~$k$-neighbor~$\colo'$ where~$S:=\Dflip(\colo,\colo')$ is connected.
Since we can enumerate all connected vertex sets~$S$ of size at most~$k$ in~$G$ in $\Oh(e^k\cdot (\Delta(G)-1)^k \cdot k \cdot n)$~time~\cite{KS21,KS15} and we can compute for each such set~$S$ a~$c$-coloring~$\colo'$ with~$\Dflip(\colo,\colo')\subseteq S$ that maximizes~$\omega(E(\colo'))$ in $\Oh(3^k\cdot c\cdot k^2 + \Delta(G) \cdot k)$~time due to~\Cref{dp}, we obtain the stated running time.  
\end{proof}
Alternatively, one can also find for each given candidate set~$S$ of size at most~$k$ a best~$c$-colorings~$\colo'$ with~$D(\colo,\colo') = S$ by enumerating all such colorings.
Note that for each candidate set~$S$ of size at most~$k$, this can be done in $\Oh((c-1)^k \cdot \Delta(G)\cdot k)$~time since each vertex~$v$ of~$S$ has to change its color to one of the colors of~$[1,c]\setminus \{\colo(v)\}$.
This implies the following even better running time for~\LMC and~\LGGC[3].

\begin{corollary}
\LMC can be solved in $\Oh(e^k \cdot (\Delta(G)-1)^{k+1} \cdot k^2 \cdot n)$~time and \LGGC[3] can be solved in~$\Oh((2\cdot e)^k \cdot (\Delta(G)-1)^{k+1} \cdot k^2 \cdot n)$~time.
\end{corollary}
\iflong
\fi

\paragraph{Hill-Climbing Algorithm}
To obtain not only a single improvement of a given coloring but a $c$-coloring with a total weight of properly colored edges as high as possible, we introduce the following hill-climbing algorithm.

Given an initial coloring~$\colo$, we set the initial value of $k$ to 1. In each step, we use the above-mentioned algorithm for \LGGC to search for an improving coloring in the $k$-flip neighborhood of the current coloring.
Whenever the algorithm finds an improving~$k$-neighbor~$\colo'$ for the current coloring~$\colo$, the current coloring gets replaced by~$\colo'$ and~$k$ gets set back to one.
If the current coloring is~$k$-optimal, the value of~$k$ is incremented and the algorithm continues to search for an improvement in the new $ k $-flip neighborhood.
This is done until a given time limit is reached.

\iflonglong
\fi
\iflong
\paragraph{ILP Formulation.}
In our experiments, we also use the following ILP formulation for~\GGC.
For each vertex~$v\in V$ and each color~$i\in [1,c]$, we use a binary variable~$x_{v,i}$ which is equal to one if and only if~$\colo'(v) = i$.
We further use for each edge~$e\in E$ a binary variable~$y_e$ to indicate whether~$e$ is properly colored with respect to~$\colo'$.
Thus, for each edge~$\{u,v\}\in E$, the variable~$y_\{u,v\}$ is set to one if and only if for each color~$i\in [1,c]$, $x_{u,i} = 0$ or~$x_{v,i} = 0$.
This is ensured by the constraint~$x_{u,i} + x_{v,i} + y_{\{u,v\}}$.
\iflonglong
\begin{align}
\intertext{maximize $\sum_{e\in E} y_e\cdot \omega(e)$ subject to}
  \sum_{i\in[1,c]} x_{v,i} = 1 &&& \text{for each vertex~$v \in V$} \label{eq:exactOneColor}\\
  x_{u,i} + x_{v,i} + y_{\{u,v\}} \leq 2 &&& \text{for each edge~$\{u,v\}\in E$}\nonumber\\
   &&& \text{with $\omega(\{u,v\}) > 0$}\nonumber\\
   &&& \text{and each color~$i\in [1,c]$} \label{eq:propColoredEdges} \\
  x_{u,i} + \sum_{j\in [1,c]\setminus \{i\}} x_{v,j} - y_{\{u,v\}} \leq 1 &&& \text{for each edge~$\{u,v\}\in E$}\nonumber\\
   &&& \text{with $\omega(\{u,v\}) < 0$}\nonumber\\
   &&& \text{and each color~$i\in [1,c]$} \label{eq:propColoredEdgesNeg} \\
  x_{v,i} \in \{0,1\} &&& \text{for each vertex~$v\in V$}\nonumber\\
   &&& \text{and each color~$i\in[1,c]$}\nonumber\\
  y_e \in \{0,1\} &&& \text{for each edge~$e\in E$}\nonumber
\end{align}

\else 
\begin{align}
\intertext{maximize $\sum_{e\in E} y_e\cdot \omega(e)$ subject to}
  \sum_{i\in[1,c]} x_{v,i} = 1 &&& \text{for each $v \in V$} \nonumber
  \\
  x_{u,i} + x_{v,i} + y_{\{u,v\}} \leq 2 &&& \text{for each $\{u,v\}\in E$}\nonumber\\
   &&& \text{with $\omega(\{u,v\}) > 0$}\nonumber\\
   &&& \text{and each~$i\in [1,c]$}\nonumber 
   \\
  x_{u,i} + \sum_{j\in [1,c]\setminus \{i\}} x_{v,j} - y_{\{u,v\}} \leq 1 &&& \text{for each $\{u,v\}\in E$}\nonumber\\
   &&& \text{with $\omega(\{u,v\}) < 0$}\nonumber\\
   &&& \text{and each~$i\in [1,c]$}\nonumber\\  
   x_{v,i} \in \{0,1\} &&& \text{for each $v\in V$}\nonumber\\
   &&& \text{and each $i\in[1,c]$}\nonumber\\
  y_e \in \{0,1\} &&& \text{for each $e\in E$}\nonumber\label{eq:propColoredEdgesNeg} 
\end{align}

\fi
\fi

Note that by adding the additional constraint~$\sum_{v\in V} x_{v, \colo(v)} \geq |V|-k$, the ILP searches for a best~$c$-coloring of  the input graph having flip-distance at most~$k$ with some initial~$c$-coloring~$\colo$.
In other words, by adding this single constraint, the ILP solves~\LGGC instead of~\GGC.

\section{Speedup Strategies}

We now introduce several speedup strategies that we use in our implementation to avoid enumerating all candidate sets. First we describe how to speed up the algorithm for \LGGC.

\subsection{Upper Bounds}
To prevent the algorithm from enumerating all possible connected subsets of size at most~$k$, we use upper bounds to determine for any given connected subset~$S'$ of size smaller than~$k$, if~$S'$ can possibly be extended to a set~$S$ of size~$k$ such that there is an improving~$c$-coloring~$\colo'$ for~$G$ where~$S$ is exactly the set of vertices~$\colo$ and~$\colo'$ do not agree on.
If there is no such possibility, we prevent our algorithm from enumerating supersets of~$S'$. 
With the next definition we formalize this concept.

\begin{definition} 
Let~$I:=(G,c,\omega,\colo,k)$ be an instance of~\LGGC and let~$S'$ with~$|S'| < k$ be a subset of vertices of~$G$. 
A value~$b(I,S')$ is an \emph{upper bound} if for each~$c$-coloring~$\colo'$ of~$G$, with~$S' \subsetneq \Dflip(\colo,\colo')$ and~$\dflip(\colo,\colo') = k$, 
\begin{align*}
b(I,S') \geq \omega(E(\colo')).
\end{align*}
\end{definition}

In our implementation, we use upper bounds as follows: Given a set~$S'$ we compute the value~$b(I,S')$ and check if it is not larger than~$\omega(E(\colo))$ for the current coloring~$\colo$. If this is the case, we abort the enumeration of supersets of~$S'$, otherwise, we continue.

We introduce two upper bounds; one for~$c=2$ and one for~$c \geq 3$. 
To describe these upper bounds, we introduce the following notation: Given a vertex~$v$ and a color~$i$, we let~$\omega_v^i := \omega (\{\{v,w\} \mid w \in N(v) , \colo(w) \neq i\})$ denote the total weight of properly colored edges incident with~$v$ if we change the color of~$v$ to~$i$ in in the~$c$-coloring~$\colo$. 
Thus, the term~$\omega^i_v - \omega^{\colo(v)}_v$ describes the improvement obtained by changing only the color of~$v$ to~$i$. 
Furthermore, let~$\omega_{\max} := \max_{e \in E} |\omega(e)|$ denote the maximum absolute edge weight.

\paragraph{Upper Bound for~$c=2$\textbf{.}} 
Let~$I$ be an instance of~\LGGC with~$c=2$ and let~$S'$ be a vertex set of size less than~$k$. 
Since~$c=2$, we let~$\overline{\colo(v)}$ denote the unique color distinct from~$\colo(v)$ for each vertex~$v$. 
For a vertex set~$A\subseteq V$, let~$\colo_{A}$ denote the coloring where~$\colo_A(v) := \colo(v)$ for all~$v\notin A$ and~$\colo_A(v) = \overline{\colo(v)}$, otherwise.
Intuitively,~$\colo_A$ is the coloring resulting from~$\colo$ when exactly the vertices in~$A$ change their colors.
For each vertex~$v \in V \setminus S'$, we define~$\alpha_v := \omega^{\overline{\colo(v)}}_v-\omega^{\colo(v)}_v + \beta_v$, where
\begin{align*}
&\beta_v :=  \sum_{e \in E(v,S') \cap  E(\colo)} 2 \cdot \omega(e) - \sum_{e \in E(v,S') \setminus E(\colo)} 2 \cdot \omega(e).
\end{align*}
Intuitively,~$\alpha_v - \beta_v$ is an upper bound for the improvement obtained when we choose to change only the color of~$v$ to~$\overline{\colo(v)}$. 
The term~$\beta_v$ corresponds to the contribution of the edges between~$v$ and the vertices of~$S'$. 
In the definition of~$\beta_v$, we take into account the edges between $v$ and $S'$ that are falsely counted twice, once when extending~$\colo_{S'}$ with~$v$ and a second time in the
term $\omega^{\overline{\colo(v)}}_v-\omega^{\colo(v)}_v$.
Hence, $\alpha_v$ is the improvement over the coloring~$\colo_{S'}$ obtained by changing only the color of~$v$.
Let~$Y \subseteq V \setminus S'$ be the~$k-|S'|$~vertices from~$V \setminus S'$ with largest~$\alpha_v$-values. 
We define the upper bound by
\begin{align*}
b_{c=2}(I,S') := \omega (E(\colo_{S'})) + \underbrace{\sum_{v \in Y} \alpha_v }_{\text{(1)}} + \underbrace{2 \binom{k-|S'|}{2} \omega_{\max}}_{\text{(2)}}.
\end{align*}
Recall that the overall goal is to find a set~$X$ such that changing the colors of~$S' \cup X$ results in a better coloring. 
The summand~(1) corresponds to an overestimation of all weights of edges incident with exactly one vertex of~$X$ by fixing the falsely counted edges between~$X$ and~$S'$ due to the included~$\beta_v$ summands. 
The summand~(2) corresponds to an overestimation of the weight of properly colored edges with both endpoints in~$X$. 
We next show that~$b_{c=2}$ is in fact an upper bound.

\begin{proposition}
If~$c=2$, then $b_{c=2}(I,S')$ is an upper bound.
\end{proposition}
\begin{proof}
Let~$\colo'$ be a coloring with~$S' \subsetneq \Dflip(\colo,\colo')$ and~$\dflip(\colo,\colo') = k$, and let~$X := \Dflip(\colo,\colo') \setminus S'$.
We show that~$\omega(E(\colo')) \leq b_{c=2}(I,S')$. 
To this end, we consider the coloring~$\colo_{S'}$ that results from~$\colo$ when exactly the vertices in~$S'$ change their colors and analyze how~$\omega(E(\colo'))$ differs from~$\omega(E(\colo_{S'}))$.

\begin{align*}
\omega(E(\colo')) =~ &\omega(E(\colo_{S'})) + \underbrace{\sum_{v \in X} (\omega^{\overline{\colo(v)}}_v - \omega^{\colo(v)}_v)}_{\text{(1)}}\\
&+ \underbrace{\sum_{\substack{e \in E(X,S') \\ e \in E(\colo)}} 2 \cdot \omega(e) - \sum_{\substack{e \in E(X,S') \\ e \not \in E(\colo)}} 2 \cdot \omega(e)}_{\text{(2)}}\\
& + \underbrace{\sum_{\substack{e \in E(X) \\ e \in E(\colo)}} 2 \cdot \omega(e) - \sum_{\substack{e \in E(X) \\ e \not \in E(\colo)}} 2 \cdot \omega(e)}_{\text{(3)}}
\end{align*}
By adding~(1) to~$\omega(E(\colo_{S'}))$, we added the weight of all properly colored edges if only~$v$ changes its color for every vertex~$v \in X$. 
The difference between~$\omega(E(\colo'))$ and~$\omega(E(\colo_{S'})) + \text{(1)}$ then consists of all edge-weights that were falsely counted in~(1) since both endpoints were moved. 
To compensate this, the summand (2) and (3) need to be added. 
Summand (2) corresponds to falsely counted edges with one endpoint in~$X$ and one endpoint in~$S'$, while (3) corresponds to falsely counted edges with both endpoints in~$X$. 
Observe that every falsely counted edge weight was counted for both of its endpoints within~$\omega(E(\colo_{S'})) + \text{(1)}$. Thus, each edge weight in~(2) and~(3) needs to be multiplied by~$2$.

Note that~(3) is upper bounded by~$2 \cdot \binom{k-|S'|}{2} \cdot \omega_{\max}$. 
Furthermore, note that~$\text{(1)}+\text{(2)}= \sum_{v \in X} \alpha_v$. 
Recall that~$Y$ consists of the~$k-|S'|$~vertices from~$V \setminus S'$ with largest~$\alpha_v$-values.
Hence, $\text{(1)}+\text{(2)} \leq \sum_{v \in Y} \alpha_v$.
This implies that~$\omega(E(\colo')) \leq b_{c=2}(I,S')$.
Consequently, $b_{c=2}(I,S')$ is an upper bound.
\end{proof}

\paragraph{Upper Bound for~$c\geq 3$\textbf{.}} We next present an upper bound~$b_{c\geq 3}$ that works for the case where~$c \geq 3$. Recall that the upper bound~$b_{c=2}$ relies on computing~$\omega(E(\colo_{S'}))$, where~$\colo_{S'}$ is the coloring resulting from~$\colo$ when exactly the vertices in~$S'$ change their colors. 
This was possible since for~$c=2$, there is only one coloring for which the flip with~$\colo$ is exactly~$S'$. 
In case of~$c \geq 3$, each vertex in~$S'$ has~$c-1 \geq 2$ options to change its color. 
Our upper bound~$b_{c\geq 3}$ consequently contains a summand $b(S')$ that overestimates the edge weights when only the vertices in~$S'$ change their colors.

To specify~$b(S')$, we introduce the following notation: 
Given a vertex~$v \in S'$ and a color~$i$, we let~$$\theta^i_v:= \omega(\{\{v,w\} \mid w \in N(v) \setminus S' , \colo(w) \neq i\}).$$ 
Analogously to~$\omega^i_v$, the value~$\theta^i_v$ describes the weight of properly colored edges when changing the color of~$v$ to~$i$, but excludes all edges inside~$S'$. 
We define the term 
\begin{align*}
b(S'):=  &~ \omega(E(\colo)) + \binom{|S'|}{2} \cdot \omega_{\max} - \sum_{\substack{ e \in E(S') \\ e \in E(\colo)}} \omega(e) \\
& + \sum_{v \in S'} \left(\max_{i \neq \colo(v)} \theta^i_v - \theta^{\colo(v)}_v\right).
\end{align*} 
As mentioned above, for~$b_{c\geq 3}$ the summand~$b(S')$ replaces the summand~$\omega(E(\colo_{S'}))$ which was used for~$b_{c=2}$. Intuitively, the sum~$\sum_{v \in S'} (\max_{i \neq \colo(v)} \theta^i_v - \theta^{\colo(v)}_v)$ is an overestimation of the improvement for properly colored edges with exactly one endpoint in~$S'$, the term~$ \binom{|S'|}{2} \cdot \omega_{\max}$ overestimates the properly colored edges inside~$S'$, and the remaining terms overestimate the properly colored edges outside~$S'$.

Analogously to~$b_{c=2}$, for each vertex~$v \in V \setminus S'$, we define a value~$\alpha_v := \max_{i \neq \colo(v)} (\omega^i_v - \omega^{\colo(v)}_v) + \beta_v$ with
$$\beta_v:= \sum_{e \in E(v,S')} 2 \cdot |\omega(e)|.$$
Again, let~$Y \subseteq V \setminus S'$ be the~$k-|S'|$ vertices with biggest~$\alpha_v$-values of~$V \setminus S'$. 
We define the upper bound by
\begin{align*}
b_{c\geq 3}(I,S') := b(S') + \sum_{v \in Y} \alpha_v  + 2 \binom{k-|S'|}{2} \cdot \omega_{\max}
\end{align*}
and show that it is in fact an upper bound.

\begin{proposition}
If~$c\geq 3$, then $b_{c\geq 3}(I,S')$ is an upper bound.
\end{proposition}

\iflong
\begin{proof}
Let~$\colo'$ be a coloring with~$S' \subsetneq \Dflip(\colo,\colo')$ and~$\dflip(\colo,\colo') = k$, and let~$X := \Dflip(\colo,\colo') \setminus S'$.
We show that~$\omega(E(\colo')) \leq b_{c \geq 3}(I,S')$. 
To this end, let~$\pF{\colo}{\colo'}{S'}$ denote the coloring that agrees with~$\colo$ on all vertices of~$V\setminus S'$ and that agrees with~$\colo'$ on all vertices of~$S'$.
To show~$\omega(E(\colo')) \leq b_{c \geq 3}(I,S')$ we analyze how~$\omega(E(\colo'))$ differs from~$\omega(E(\pF{\colo}{\colo'}{S'}))$.

\begin{align*}
\omega(E(\colo')) \leq~ &\omega(E(\pF{\colo}{\colo'}{S'})) + \underbrace{\sum_{v \in X} (\omega^{\colo'(v)}_v - \omega^{\colo(v)}_v)}_{\text{(1)}}\\
& + \underbrace{\sum_{\substack{e\in E(S',~X)}} 2 \cdot |\omega(e)|}_{\text{(2)}} + \underbrace{\sum_{\substack{e\in E(X)}} 2 \cdot |\omega(e)|}_{\text{(3)}}
\end{align*}

By adding~(1) to~$\omega(E(\pF{\colo}{\colo'}{S'}))$, we added the weight of all properly colored edges if only~$v$ changes its color for every vertex~$v \in X$. The difference between~$\omega(E(\colo'))$ and~$\omega(E(\pF{\colo}{\colo'}{S'})) + \text{(1)}$ then consists of all edge-weights that were falsely counted in~(1) since both endpoints were moved. 
To compensate this, the summand (2) and (3) were added. 
Summand (2) overestimates the weight of falsely counted edges with one endpoint in~$X$ and one endpoint in~$S'$, while (3) overestimates the weight of falsely counted edges with both endpoints in~$X$. Observe that every falsely counted edge weight may be counted for both of its endpoints within~$\omega(E(\colo|_{S'}))+\text{(1)}$. 
Thus, each edge weight in~(2) and~(3) needs to be multiplied by~$2$.

Note that~$\text{(1)}+\text{(2)} \leq \sum_{v \in X} \alpha_v \leq \sum_{v \in Y} \alpha_v$ and that~$\text{(3)} \leq 2 \binom{k-|S'|}{2} \cdot \omega_{\max}$. Therefore, it remains to show that~$\omega(E(\pF{\colo}{\colo'}{S'})) \leq b(S')$. 
To this end, note that~$\omega(E(\pF{\colo}{\colo'}{S'}))$ can be expressed by the sum of~$\omega(E(\colo))$, improvement of the weight of properly colored edges inside~$S'$ between~$\colo$ and~$\pF{\colo}{\colo'}{S'}$, and~$\sum_{v \in S'}  (\theta^{\colo'(v)}_v - \theta^{\colo(v)}_v)$:

\begin{align*}
\omega(E(\pF{\colo}{\colo'}{S'})) =~& \omega(E(\colo)) + \sum_{\substack{ e \in E(S') \\ e \in E(\pF{\colo}{\colo'}{S'})}} \omega(e) - \sum_{\substack{ e \in E(S') \\ e \in E(\colo)}} \omega(e) \\ 
&+ \sum_{v \in S'}  (\theta^{\colo'(v)}_v - \theta^{\colo(v)}_v).
\end{align*}

Since~$\binom{|S'|}{2} \cdot \omega_{\max}$ is at least as big as the sum of the weights of properly edges inside~$S'$ under~$\pF{\colo}{\colo'}{S'}$, we conclude~$\omega(E(\pF{\colo}{\colo'}{S'})) \leq b(S')$.
Hence,~$b_{c \geq 3}$ is an upper bound.
\end{proof}
\fi

\subsection{Prevention of Redundant Flips
}

We introduce further speed-up techniques that we used in our implementation of the hill-climbing algorithm. Roughly speaking, the idea behind these speed-up techniques is to exclude vertices that are not contained in an improving flip~$\Dflip(\colo,\colo')$ of any~$k$-neighbor~$\colo'$ of~$\colo$. To this end, we introduce for each considered value of~$k$ an \emph{auxiliary vertex set}~$V_k$ containing all remaining vertices that are potentially part of an improving flip of a~$k$-neighbor of~$\colo$.
 For each value of~$k$, the set~$V_k$ is initialized once with~$V$, when we search for the first time for an improving~$k$-neighbor.

It is easy to see that all vertices~$x$ that are~$(i,k)$-blocked for all~$i \neq \colo(x)$ can be removed from~$V_k$ if each edge of~$G$ has weight 1. 
This also holds for general instances when considering an extension of the definition of~$(i,k)$-blocked vertices for arbitrary weight functions.
Moreover, whenever our algorithm has verified that a vertex~$v$ is in no improving $k$-flip~$\Dflip(\colo,\colo')$, then we may remove~$v$ from~$V_k$.

Recall that we set the initial value of $k$ to one and increment~$k$ if the current coloring~$\colo$ is $k$-optimal. 
If at any time our algorithm replaces the current coloring~$\colo$ by a better coloring~$\colo'$, we set~$k$ back to one and continue by searching for an improving~$k$-neighbor of the new coloring~$\colo'$, where~$k$ again is incremented if necessary. 
Now, for each value of~$k'$ that was already considered for a previous coloring, we only consider the remaining vertices of~$V_{k'}$ together with vertices that have a small distance to the flip between $\colo'$ and the last previously encountered~$(k'-1)$-optimal coloring. 
This idea is formalized by the next lemma.


\begin{lemma} \label{lem:redundantFlips}
Let~$G=(V,E)$ be a graph, let~$\omega:E\to \mathds{Q}$ be an edge-weight function, and let~$k$ be an integer.
Moreover, let~$\colo$ and~$\colo'$ be~$(k-1)$-optimal~$c$-colorings of~$G$ and let~$v$ be a vertex within distance at least~$k+1$ to each vertex of~$\Dflip(\colo,\colo')$.
If there is no improving~$k$-neighbor~$\widehat{\colo}$ of~$\colo$ with~$v\in \Dflip(\colo,\widehat{\colo})$, then there is no improving~$k$-neighbor~$\widetilde{\colo}$ of~$\colo'$ with~$v\in \Dflip(\colo',\widetilde{\colo})$.
\end{lemma}

\iflong
\begin{proof}
We prove the lemma by contraposition. Let~$\widetilde{\colo}$ be an improving~$k$-neighbor of~$\colo'$ with~$v \in \Dflip(\colo',\widetilde{\colo})$. We show that there is an improving~$k$-neighbor~$\widehat{\colo}$ of~$\colo$ with~$v \in \Dflip(\colo,\widehat{\colo})$.

The~$c$-coloring~$\widehat{\colo}$ agrees with~$\widetilde{\colo}$ on all vertices of~$\Dflip(\colo',\widetilde{\colo})$ and agrees with~$\colo$ on all other vertices of~$V$.
Hence, $\Dflip(\colo, \widehat{\colo})$ contains the vertex~$v$. 
Moreover, $\widehat{\colo}$ and~$\colo$ disagree on at most~$\dflip(\colo',\widetilde{\colo}) \leq k$ positions which implies that~$\widehat{\colo}$ is a~$k$-neighbor of~$\colo$. 
It remains to show that~$\widehat{\colo}$ improves over~$\colo$. 
To this end, we analyze the edge set~$X \subseteq E$ of all edges with at least one endpoint in~$\Dflip(\colo',\widetilde{\colo})$. 
Consider the following claim about properly colored edges.

\begin{claim} \label{Claim: Colored Edges for Lemma}
It holds that
\begin{enumerate}
\item[$a)$] $E(\widetilde{\colo}) \setminus X = E(\colo') \setminus X$ and $E(\widehat{\colo}) \setminus X = E(\colo) \setminus X$ ,
\item[$b)$]$ E(\widetilde{\colo}) \cap X = E(\widehat{\colo}) \cap X$ and $E(\colo) \cap X = E(\colo') \cap X$.
\end{enumerate}
\end{claim}

\begin{claimproof} \normalfont
$a)$ 
Let~$e$ be an edge of~$E \setminus X$. 
Note that both endpoints of~$e$ are elements of~$V \setminus \Dflip(\colo',\widetilde{\colo})$. 
Thus, the endpoints of~$e$ have distinct colors under~$\widetilde{\colo}$ if and only if they have distinct colors under~$\colo'$. 
Consequently, $E(\widetilde{\colo}) \setminus X = E(\colo') \setminus X$. 
Furthermore, by definition of~$\widehat{\colo}$, $\Dflip(\colo,\widehat{\colo})=\Dflip(\colo',\widetilde{\colo})$, which implies that~$E(\widehat{\colo}) \setminus X = E(\colo) \setminus X$.

$b)$ 
Since~$\colo'$ is~$(k-1)$-optimal and~$\widetilde{\colo}$ is an improving~$k$-neighbor of~$\colo'$, the set~$\Dflip(\colo',\widetilde{\colo})$ contains exactly~$k$ vertices. 
Thus, we may assume by Lemma~\ref{lem connected} that~$\Dflip(\colo',\widetilde{\colo})$ is connected. 
Consequently, each vertex of~$\Dflip(\colo',\widetilde{\colo})$ has distance at most~$k-1$ from~$v$, since~$v$ is contained in~$\Dflip(\colo',\widetilde{\colo})$.

Let~$e$ be an edge of~$X$. 
Since each vertex of~$\Dflip(\colo',\widetilde{\colo})$ has distance at most~$k-1$ from~$v$, both endpoints of~$e$ have distance at most~$k$ from~$v$. 
Together with the fact that~$v$ has distance at least~$k+1$ from~$\Dflip(\colo,\colo')$, this implies that no endpoint of~$e$ is contained in~$\Dflip(\colo,\colo')$. 
Hence, $\colo$ and~$\colo'$ agree on both endpoints of~$e$. 
Consequently, the edge~$e$ is properly colored under~$\colo$ if and only if~$e$ is properly colored under~$\colo'$.
This then implies that~$E(\colo) \cap X = E(\colo') \cap X$.

By the definition of~$\widehat{\colo}$ and the fact that~$\colo$ and~$\colo'$ agree on the endpoints of each edge of~$X$, $\widehat{\colo}$ and~$\widetilde{\colo}$ agree on the endpoints of each edge of~$X$.
This then implies that~$E(\widetilde{\colo}) \cap X = E(\widehat{\colo}) \cap X$.  $\hfill \diamond$
\end{claimproof}

We next use~\Cref{Claim: Colored Edges for Lemma} to show that~$\widehat{\colo}$ is an improving neighbor of~$\colo$. 
Since~$\widetilde{\colo}$ is an improving neighbor of~$\colo'$ we have~$\omega(E(\widetilde{\colo})) > \omega(E(\colo'))$, which implies
\begin{align*}
\omega(E(\widetilde{\colo}) \cap X) + \omega(E(\widetilde{\colo}) \setminus X) > \omega(E(\colo') \cap X) + \omega(E(\colo') \setminus X).
\end{align*}
Together with~\Cref{Claim: Colored Edges for Lemma}~$a)$, we then have
\begin{align*}
\omega(E(\widetilde{\colo}) \cap X) > \omega(E(\colo') \cap X).
\end{align*} 
Moreover, due to~\Cref{Claim: Colored Edges for Lemma}~$b)$, we have
\begin{align*}
\omega( E(\widehat{\colo}) \cap X) > \omega(E(\colo) \cap X).
\end{align*}
Finally, since~$E(\widehat{\colo}) \setminus X = E(\colo) \setminus X$ by~\Cref{Claim: Colored Edges for Lemma}~$a)$, we may add the weights of all edges in~$E(\widehat{\colo}) \setminus X$ to the left side of the inequality and the weight of all edges in~$E(\colo) \setminus X$ to the right side. 
We end up with the inequality~$\omega(E(\widehat{\colo})) > \omega(E(\colo))$ which implies that~$\widehat{\colo}$ is an improving~$k$-neighbor of~$\colo$.
\end{proof}
\fi
\iflong
We next describe how we exploit~\Cref{lem:redundantFlips} in our implementation: 
We start with a coloring~$\colo$ and search for improving~$k$-neighbors of~$\colo$ for increasing values of~$k$ starting with~$k=1$. 
Whenever we find an improving neighbor~$\colo'$ of~$\colo$ we continue by searching for an improving neighbor~$\colo''$ of~$\colo'$ starting with~$k=1$ again. 
\fi
\iflong We \fi use~\Cref{lem:redundantFlips} as follows: if we want to find an improving~$k$-neighbor for a~$(k-1)$-optimal coloring~$\colo'$, we take the last previously encountered~$(k-1)$-optimal coloring~$\colo$ and add only the vertices to~$V_k$ that have distance at most~$k$ from~$\Dflip(\colo,\colo')$, instead of setting~$V_k$ back to~$V$.
This is correct since every vertex which is not in~$V_k$, is not part of any improving~$k$-flip of $\colo$ and therefore according to~\Cref{lem:redundantFlips}, the only vertices outside of~$V_k$ that can possibly be in an improving~$k$-flip of~$\colo'$ are those with distance at most~$k$ from~$\Dflip(\colo,\colo')$.



\iflonglong
\fi

Next, we provide a further technique to identify vertices that can be removed from~$V_k$. The idea behind this technique can be explained as follows: if a vertex can be excluded from~$V_k$, then all equivalent vertices can also be excluded, where equivalence is defined as follows.

\begin{definition}
Let~$G=(V,E)$ be a graph, let~$\omega:E\to \mathds{Q}$ be an edge-weight function.
Two vertices~$v$ and~$w$ of~$G$ are \emph{weighted twins} if $N(v)\setminus \{w\} = N(w)\setminus \{v\}$ and~$\omega(\{v,x\}) = \omega(\{w,x\})$ for each~$x\in N(v)\setminus \{w\}$.
\end{definition}


\begin{lemma}
Let~$G=(V,E)$ be a graph, let~$\omega:E\to \mathds{Q}$ be an edge-weight function, and let~$k$ be an integer.
Moreover, let~$\colo$ be a~$c$-coloring of~$G$ and let~$v$ and~$w$ be weighted twins in~$G$ with~$\colo(v) = \colo(w)$.
If there is no improving~$k$-neighbor~$\colo'$ of~$\colo$ with~$v\in \Dflip(\colo,\colo')$, then there is no improving~$k$-neighbor~$\widetilde{\colo}$ of~$\colo$ with~$w\in \Dflip(\colo,\widetilde{\colo})$.
\end{lemma}
\iflong
\begin{proof}
Assume towards a contradiction that there is an improving~$k$-neighbor~$\widetilde{\colo}$ of~$\colo$ with~$w\in \Dflip(\colo,\widetilde{\colo})$.
By assumption,~$v\notin \Dflip(\colo,\widetilde{\colo})$.
Let~$\colo'$ be the~$c$-coloring that agrees with~$\widetilde{\colo}$ on~$V\setminus \{v,w\}$ and where~$\colo'(v) := \widetilde{\colo}(w)$ and~$\colo'(w) := \widetilde{\colo}(v) = \colo(v)$.
Recall that~$\omega(\{v,x\}) = \omega(\{w,x\})$ for each~$x\in N(v) \cap N(w)$ and that~$N(v) \setminus \{w\} = N(w) \setminus \{v\}$. 
For each~$x\in N(v) \cap N(w)$, let~$E_x := \{\{v,x\}, \{w,x\}\}$.
Note that since~$\widetilde{\colo}(v) \neq \widetilde{\colo}(w)$, at least one edge of~$E_x$ is contained in~$E(\widetilde{\colo})$.
If both edges of~$E_x$ are contained in~$E(\widetilde{\colo})$, then~$\colo'(x)=\widetilde{\colo}(x) \not\in \{\colo'(v), \colo'(w)\}$ and thus both edges of~$E_x$ are contained in~$E(\colo')$.
If only one edge of~$E_x$ is contained in~$E(\widetilde{\colo})$, then~$E(\colo')$ contains exactly the other edge of~$E_x$ since~$\colo'(v) = \widetilde{\colo}(w)$,~$\colo'(w) = \widetilde{\colo}(v)$, and~$\colo'(x) = \widetilde{\colo}(x)$.
Since the weight of both edges of~$E_x$ are the same for each~$x\in N(v) \cap N(w)$, we have~$\omega(E(\colo')) = \omega(E(\widetilde{\colo}))$.
Hence,~$\colo'$ is an improving~$k$-neighbor of~$\colo$ with~$v\in \Dflip(\colo,\colo')$, a contradiction.
\end{proof}
\fi
Consequently, when our algorithm removes a vertex~$v$ from~$V_k$ for some~$k$ because no improving~$k$-neighbor~$\colo'$ of~$\colo$ contains~$v$, then it also removes all weighted twins of~$v$ with the same color as~$v$ from~$V_k$.

\begin{table}[t]
\caption{The graphs from the G-set for which LS or ILP found an improved coloring, or for which we verified that MOH colorings are optimal (for $c=3$).
MOH shows the value of the published solutions of~\protect\cite{MH17}, LS and ILP show the best solution of our hill-climbing algorithm and any of the two ILP-runs, respectively. The best coloring is bold.
Finally, UB shows the better upper bound computed during the two ILP-runs.
For empty entries, no improved coloring was found.
For bold UB entries, some found solution matches this upper bound,  verifying its optimality.
}
\centering
{\tiny
\tabMiniGGC
  
}
\label{tabmini}
\end{table}
\section{Implementation and Experimental Results}
Our hill-climbing algorithm (LS) is implemented in JAVA/Kotlin and uses the graph library JGraphT. 
To enumerate all connected candidate sets, we use a JAVA implementation of a polynomial-delay algorithm for enumerating all connected induced subgraphs of a given size~\cite{KS21}.
   
We used the graphs from the G-set benchmark\footnote{\url{https://web.stanford.edu/~yyye/yyye/Gset/}}, an established benchmark data set for~\GGC with~$c\in \{2,3,4\}$ (and thus also for \MC{})~\cite{BH13,FPRR02,MH17,S+15,Wang13,ZLA13}. The data set consists of 71 graphs with vertex-count between 800 and 20,000 and a density between 0.02\% and 6\%.

As starting solutions, we used the solutions computed by the MOH algorithm of~Ma~and~Hao~\cite{MH17} for each graph of the G-set and each~$c\in \{2,3,4\}$. For~$c=3$ and~$c=4$, these are the best known solutions for all graphs of the G-set. MOH is designed to quickly improve substantially on starting solutions but after a while progress stalls~(we provide more details on this below). In contrast, our approach makes steady progress but is not as fast as MOH concerning the initial improvements, as preliminary experiments showed.
Hence, we focus on evaluating the performance of LS as a post-processing for MOH by trying to improve their solutions quickly.

\begin{table}
 \caption{The number of instances where LS or ILP found improved solutions. Column `improvable' shows how many best known MOH colorings~\protect\cite{MH17}  \emph{might} be suboptimal (as they do not meet the ILP upper bounds). Columns LS and ILP show how many of these solutions where improved by the respective approaches.
 Columns~I$_1$,~I$_2$, and~I$_3$ show for how many instances the first improvement was found by LS within 10 seconds, between 10 and 60 seconds, and after more than 60 seconds, respectively.}
 \centering
 \improvementType
\label{tabImp}
\end{table}

\begin{table*}
\caption{The solutions of the best found~$c$-coloring for any of the G-set graphs for~$c=2$.
The column MOH shows the value of the published solutions of~Ma~and~Hao~\cite{MH17}. 
}
\centering
{\tiny
\tabTwo
}
\label{tab:2}
\end{table*}

\begin{table*}
\caption{The solutions of the best found~$c$-coloring for any of the G-set graphs for~$c=3$.}
\centering
{\fontsize{7}{7} \selectfont
\tabThree
}
\label{tab:3}
\end{table*}

\begin{table*}
\caption{The solutions of the best found~$c$-coloring for any of the G-set graphs for~$c=4$.}
\centering
{\fontsize{7}{7} \selectfont
\tabFour
}
\label{tab:4}
\vspace{10pt}
\end{table*}

\begin{table*}
\caption{For each value of~$k$, the number of instances for which an improving flip of size exactly~$k$ was found.}
\centering
\largestType
\label{tab:largestK}
\end{table*}

\iflong For one graph (g23) and each~$c\in \{2,3,4\}$, there is a large gap between the value of the published coloring and the stated value of the corresponding coloring\iflong{} (for example, for~$c=3$, the published coloring has a value of~$13\,275$ whereas it is stated that the coloring has a value of~$17\,168$)\fi.
To not exploit this gap in our evaluation, we only considered the remaining 70 graphs. \fi
\iflong These \fi 70 graphs are of two types: 34 graphs are unit graphs (where each edge has weight 1) and 36 graphs are signed graphs (where each edge has \iflong either \fi weight 1 or -1). For each of these graphs, we ran experiments for each~$c\in \{2,3,4\}$ with a time limit of 30 minutes and the published MOH solution as initial solution. 
In addition to LS, for each instance we ran standard ILP-formulations for \GGC (again for 30 minutes) using the Gurobi solver version 9.5, once without starting solution and once with the MOH solution as starting solution. 
Each run of an ILP provides both a best found solution and an upper bound on the maximum value of any~$c$-coloring for the given instance.
Each experiment was performed on a single thread of an Intel(R) Xeon(R) Silver 4116 CPU with 2.1 GHz, 24~CPUs and 128 GB RAM.

The ILP upper bounds verified the optimality of 22 MOH solutions. Thus, of the 210 instances, only 188 instances are interesting in the sense that LS or the ILP can find an improved solution. The upper bounds also verified the optimality of 8 further improved solutions found by LS or ILP.

In total, the ILP found better colorings than the MOH coloring for 43 of the 188 instances
.
In comparison, \iflong our hill-climbing algorithm \fi was able to improve on the MOH solutions for 69 instances of the 188 instances.
\Cref{tabmini} gives the results for~$c=3$, showing those instances where the MOH coloring was verified to be optimal by the ILP or where LS or the ILP found an improved coloring.
\iflong The full overview for~$c\in \{2,3,4\}$ is shown in~\Cref{tab:2,tab:3,tab:4}.\fi

Over all~$c\in \{2,3,4\}$, on 35 instances, both LS and the ILP found improved colorings compared to the MOH coloring. 
%
For~$c>2$, both approaches find new record colorings. More precisely, for 23 instances, only the ILP found a new record coloring; for 6 instances, both approaches found a new record coloring, and for 38 instances only LS found a new record coloring. Thus, LS finds improvements also for very hard instances on which MOH provided the best known solutions so far.


The MOH solutions were obtained within a time limit of 30, 120, and 240 minutes for small, medium, and large instances, respectively. 
Each such run was repeated at least 10 times.
The average time MOH took to find the best solution was 33\% of the respective time limit. 
Hence, on average, after MOH found their best solution, in the remaining time (at least 20 minutes), MOH did not find any better solution. 
For all instances where LS was able to improve on the MOH solution, the average time to find the first improving flip was~$15.17$ seconds.
 Table~\ref{tabImp} shows an overview on the number of improved instances and the time when LS found the first improvement. 
 It is also interesting to see for which value of~$k$ the first improvement was found (in other words, the smallest value~$k$ such that the MOH solutions are not~$k$-flip optimal). 
 Table~\ref{tab:smallestK} shows for how many instances which value of~$k$ was the smallest to obtain an improvement. 
 On average, this value of~$k$ was~$3.39$. 
 Hence, it is indeed helpful to consider larger values of~$k$ than the commonly used values of 1 or 2.

\begin{table*}
\caption{For each value of~$k$, the number of instances for which the first improving flip that was found had size exactly~$k$.}
\centering
\firstType
\label{tab:smallestK}
\end{table*}

 We summarize our main experimental findings as follows. 
 First, parameterized local search can be used successfully as a post-processing for state-of-the-art heuristics for \GGC, in many cases leading to new record solutions for~$c>2$ (see~\Cref{tab:3,tab:4}). 
 Second, the number of instances where an improvement was found is larger for LS than for the ILP approaches. 
 Third, to find improved solutions, it is frequently necessary to explore $k$-flip neighborhoods for larger values of~$k$ (see~\Cref{tab:largestK,tab:smallestK}). 
 Finally, this can be done within an acceptable amount of time by using our algorithm for \LGGClong and our speed-up strategies. 

\section{Conclusion}
 
In this work we analyzed~\LGGC from both a theoretical and practical point of view.
Form a negative point of view, we showed that both the strict and the permissive version of~\LGGC cannot be solved in $f(k) \cdot n^{o(k)}$~time for any computable function~$f$, unless the ETH fails.
From a positive point of view, we presented an algorithm that solves these problems in $\Delta^{\Oh(k)} \cdot n^{\Oh(1)}$~time.
Moreover, we implemented this algorithm and evaluated its performance as a post-processing for a state-of-the-art heuristic for \GGC.
Our experimental findings indicate that parameterized local search might be a promising technique in the design of local search algorithms and that its usefulness should be explored for further hard problems, in particular as post-processing for state-of-the-art heuristics to improve already good solutions.

\subparagraph{Open questions.}
Our results in this work leave several questions open and give raise to new research directions.
From a practical point of view, it would be interesting to consider a combined implementation of the MoH algorithm with our hill-climbing algorithm based on the~$k$-flip neighborhood.
In such a combined implementation, one could for example consider two different time limits~$t_1$ and~$t_2$: 
First, until time limit~$t_1$ is reached, we let the MoH-algorithm run.
Afterwards, we take the best found solution by MoH as starting solution for our hill-climbing algorithm which we then run until time limit~$t_2$ is reached.
It would be interesting to analyze the final solution quality with respect to these two time limits.
In other words, it would be interesting to analyze when switching from MoH to our hill-climbing algorithm is promising.
Such an approach could also be considered with respect to combinations of other state-of-the-art heuristics for~\GGC.
For example for~$c=2$, that is, for~\MC, one could consider analyzing the usefulness of our hill-climbing algorithm as a post-processing algorithm for algorithms like TS-UBQP~\cite{KHLWG13} or TSHEA~\cite{WWL15}.

Regarding the practical evaluation of parameterized local search algorithms, we believe that some of the techniques introduced in this chapter might find successful applications also for other problems.
In particular, the technique we introduced to prevent redundantly checking candidates (see~\Cref{lem:redundantFlips}) seems promising: 
We can ignore candidates that contain a vertex~$v$ for which no vertex of distance~$\Oh(k)$ has changed since the last time we verified that~$v$ is not contained in any improving candidate.
This technique was already successfully adapted in a local search solver for~\textsc{Weighted Vertex Cover}~\cite{Ull23} and we believe that it might find successful application for local search versions of problems like~\textsc{Max SAT}~\cite{Szei11} or \textsc{Hitting Set}.

From a theoretical point of view, one could ask for a smaller parameter in the basis of the worst-case running time.
For example, one may ask whether we can replace the maximum degree in the basis by the~$h$-index of the input graph, that is, if we can solve~\GGC in $h^{\Oh(k)} \cdot n^{\Oh(1)}$~time.
In the aim of developing an algorithm with such a running time, one could for example branch into all possible ways to flip up to~$k$ high-degree vertices.
For each such branch, one then needs to find a coloring that improves over the initial coloring that only flips low-degree vertices.
This would then necessitate solving a gap-version of~\GGC similar to the one introduced by Komusiewicz and Morawietz~\cite{KM22} for the local search version of~\textsc{Vertex Cover}.

Based on the~\W1-hardness results for~\textsc{LS Min Bisection} and~\textsc{LS Max Bisection} with respect to the search radius~$k$, one might also consider revisiting the parameterized complexity of these problems with respect to~$k$ plus some additional parameter~$\ell$.
So far, the only known algorithm for~\textsc{LS Min Bisection} and~\textsc{LS Max Bisection} run in FPT-time with respect to~$k$ on graphs on bounded local treewidth~\cite{FFL+12}. 
These algorithms imply that~\textsc{LS Min Bisection} and~\textsc{LS Max Bisection} can be solved in $2^{\Delta^{\Oh(k)}}\cdot n^{\Oh(1)}$~time, but the existence of algorithms for these problems that run in $\Delta^{\Oh(k)} \cdot n^{\Oh(1)}$~time are open.
Based on the restriction that each part of the partition is equally-sized in any solution of these problems, one might need to consider candidates to flip that are not connected in the input graph to achieve a better solution.

\section*{Acknowledgements}
Nils Morawietz was supported by the Deutsche Forschungsgemeinschaft (DFG), project OPERAH, KO~3669/5-1.

\section*{Contribution Statement}
The research for this article was done while all authors were members of Philipps-Universität Marburg, Department of Mathematics and Computer Science.

\bibliographystyle{plain}
\bibliography{my_bib}
\end{document}

\iflong

\newpage

\begin{table*}
\caption{The solutions of the best found~$c$-coloring for any of the G-set graphs for~$c=2$.
The column MOH shows the value of the published solutions of~Ma~and~Hao~\cite{MH17}. 
The column LS shows the best solution our hill-climbing algorithm found.
The column ILP shows the best solution of any of the two  ILP-runs.
Finally, column UB shows the best upper bound for any solution of the corresponding graph our ILP was able to find.
If there is no number in the LS or ILP column for an instance, then the corresponding algorithm was not able to find a better solution than the published solutions of~Ma~and~Hao~\cite{MH17}. 
If the UB entry is bold for some instance, then some found solution matches this upper bound, that is, the optimality of the best found solution was verified.}
\centering
{\tiny
\tabTwo
}
\label{tab:2}
\end{table*}

\begin{table*}
\caption{The solutions of the best found~$c$-coloring for any of the G-set graphs for~$c=3$.}
\centering
{\fontsize{7}{7} \selectfont
\tabThree
}
\label{tab:3}
\end{table*}

\begin{table*}
\caption{The solutions of the best found~$c$-coloring for any of the G-set graphs for~$c=4$.}
\centering
{\fontsize{7}{7} \selectfont
\tabFour
}
\label{tab:4}
\vspace{10pt}
\end{table*}

\begin{table*}
\caption{For each instance, the largest values of~$k$ for which an improving flip was found.}
\centering
\largestType
\label{tab:largestK}
\end{table*}

\begin{table*}
\caption{The values of~$k$ for which the first improving flip was found.}
\centering
\firstType
\label{tab:smallestK}
\end{table*}

\fi
\end{document}

%% file: tabmini.tex
\newcommand{\tabMiniGGC}{
  \begin{tabular}{lrr|r|r||rr}
data &$|V|$&$|E|$& MOH & LS & ILP & UB  \\\hline
g11 & 800 & 1\,600 &  669 &  ---  & \bfseries  671 & \bfseries  671  \\
g12 & 800 & 1\,600 &  660 & 661 & \bfseries  663 & \bfseries  663  \\
g13 & 800 & 1\,600 &  686 & 687 & \bfseries  688 & \bfseries  688  \\
g15 & 800 & 4\,661 &  3\,984 & \bfseries 3\,985 & \bfseries  3\,985 &  4\,442  \\
g24 & 2\,000 & 19\,990 &  17\,162 & \bfseries 17\,163 &  --- &  19\,989  \\
g25 & 2\,000 & 19\,990 &  17\,163 & \bfseries 17\,164 &  --- &  19\,989  \\
g26 & 2\,000 & 19\,990 &  17\,154 & \bfseries 17\,155 &  --- &  19\,989  \\
g27 & 2\,000 & 19\,990 &  4\,020 & \bfseries 4\,021 &  --- &  9\,840  \\
g28 & 2\,000 & 19\,990 &  3\,973 & \bfseries 3\,975 &  --- &  9\,822  \\
g31 & 2\,000 & 19\,990 &  4\,003 & \bfseries 4\,005 &  --- &  9\,776  \\
g32 & 2\,000 & 4\,000 &  1\,653 & 1658 & \bfseries  1\,666 &  1\,668  \\
g33 & 2\,000 & 4\,000 &  1\,625 & 1628 & \bfseries  1\,636 &  1\,640  \\
g34 & 2\,000 & 4\,000 &  1\,607 & 1609 & \bfseries  1\,616 &  1\,617  \\
g35 & 2\,000 & 11\,778 &  10\,046 & \bfseries 10\,048 &  --- &  11\,711  \\
g37 & 2\,000 & 11\,785 &  10\,052 & \bfseries 10\,053 & \bfseries  10\,053 &  11\,691  \\
g40 & 2\,000 & 11\,766 &  2\,870 & \bfseries 2\,871 &  --- &  5\,471  \\
g41 & 2\,000 & 11\,785 &  2\,887 & \bfseries 2\,888 &  --- &  5\,452  \\
g48 & 3\,000 & 6\,000 & \bfseries  6\,000 &  ---  &  --- & \bfseries  6\,000  \\
g49 & 3\,000 & 6\,000 & \bfseries  6\,000 &  ---  &  --- & \bfseries  6\,000  \\
g50 & 3\,000 & 6\,000 & \bfseries  6\,000 &  ---  &  --- & \bfseries  6\,000  \\
g55 & 5\,000 & 12\,498 &  12\,427 & 12\,429 & \bfseries  12\,432 &  12\,498  \\
g56 & 5\,000 & 12\,498 &  4\,755 & \bfseries 4\,757 &  --- &  6\,157  \\
g57 & 5\,000 & 10\,000 &  4\,080 & 4092 & \bfseries  4\,103 &  4\,154  \\
g59 & 5\,000 & 29\,570 &  7\,274 & \bfseries 7\,276 &  --- &  14\,673  \\
g61 & 7\,000 & 17\,148 &  6\,858 & \bfseries 6\,861 &  --- &  8\,728  \\
g62 & 7\,000 & 14\,000 &  5\,686 & \bfseries 5\,710 &  5\,706 &  5\,981  \\
g63 & 7\,000 & 41\,459 &  35\,315 & \bfseries 35\,318 &  --- &  41\,420  \\
g64 & 7\,000 & 41\,459 &  10\,429 & \bfseries 10\,437 &  --- &  20\,713  \\
g65 & 8\,000 & 16\,000 &  6\,489 & 6\,512 & \bfseries  6\,535 &  6\,711  \\
g66 & 9\,000 & 18\,000 &  7\,414 & 7\,442 & \bfseries  7\,443 &  7\,843  \\
g67 & 10\,000 & 20\,000 &  8\,088 & 8\,116 & \bfseries  8\,141 &  9\,080  \\
g70 & 10\,000 & 9\,999 & \bfseries  9\,999 &  ---  &  --- & \bfseries  9\,999  \\
g72 & 10\,000 & 20\,000 &  8\,190 & 8\,224 & \bfseries  8\,244 &  9\,166  \\
g77 & 14\,000 & 28\,000 &  11\,579 & \bfseries 11\,632 &  11\,619 &  13\,101  \\
g81 & 20\,000 & 40\,000 &  16\,326 & \bfseries 16\,392 &  16\,374 &  18\,337  \\
\end{tabular}
}

%% file: tab2.tex
\newcommand{\tabTwo}{
\begin{tabular}{lrr||r|r||rr||rr||rr}
data &n&m & MOH & LS & ILP & UB& ILP$_{1}$ & UB$_{1}$& ILP$_{2}$ & UB$_{2}$  \\\hline
g1 & 800 & 19176 & \bfseries  11624 & --- &  --- &  16188 &  --- &  16188 &  --- &  16225 \\
g2 & 800 & 19176 & \bfseries  11620 & --- &  --- &  15915 &  --- &  15915 &  --- &  16254 \\
g3 & 800 & 19176 & \bfseries  11622 & --- &  --- &  15766 &  --- &  15766 &  --- &  16058 \\
g4 & 800 & 19176 & \bfseries  11646 & --- &  --- &  16059 &  --- &  16217 &  --- &  16059 \\
g5 & 800 & 19176 & \bfseries  11631 & --- &  --- &  16182 &  --- &  16182 &  --- &  16220 \\
g6 & 800 & 19176 & \bfseries  2178 & --- &  --- &  6387 &  --- &  6627 &  --- &  6387 \\
g7 & 800 & 19176 & \bfseries  2006 & --- &  --- &  6225 &  --- &  6339 &  --- &  6225 \\
g8 & 800 & 19176 & \bfseries  2005 & --- &  --- &  6209 &  --- &  6209 &  --- &  6215 \\
g9 & 800 & 19176 & \bfseries  2054 & --- &  --- &  6106 &  --- &  6106 &  --- &  6379 \\
g10 & 800 & 19176 & \bfseries  2000 & --- &  --- &  6315 &  --- &  6315 &  --- &  6322 \\
g11 & 800 & 1600 & \bfseries  564 & --- &  --- & \bfseries  564 &  --- & \bfseries  564 &  --- & \bfseries  564 \\
g12 & 800 & 1600 & \bfseries  556 & --- &  --- & \bfseries  556 &  --- & \bfseries  556 &  --- & \bfseries  556 \\
g13 & 800 & 1600 & \bfseries  582 & --- &  --- & \bfseries  582 &  --- & \bfseries  582 &  --- & \bfseries  582 \\
g14 & 800 & 4694 & \bfseries  3064 & --- &  --- &  3158 &  --- &  3158 &  --- &  3158 \\
g15 & 800 & 4661 & \bfseries  3050 & --- &  --- &  3139 &  --- &  3148 &  --- &  3139 \\
g16 & 800 & 4672 & \bfseries  3052 & --- &  --- &  3144 &  --- &  3144 &  --- &  3148 \\
g17 & 800 & 4667 & \bfseries  3047 & --- &  --- &  3144 &  --- &  3144 &  --- &  3153 \\
g18 & 800 & 4694 & \bfseries  992 & --- &  --- &  1137 &  --- &  1140 &  --- &  1137 \\
g19 & 800 & 4661 & \bfseries  906 & --- &  --- &  1044 &  --- &  1046 &  --- &  1044 \\
g20 & 800 & 4672 & \bfseries  941 & --- &  --- &  1069 &  --- &  1074 &  --- &  1069 \\
g21 & 800 & 4667 & \bfseries  931 & --- &  --- &  1072 &  --- &  1074 &  --- &  1072 \\
g22 & 2000 & 19990 & \bfseries  13359 & --- &  --- &  17677 &  --- &  17888 &  --- &  17677 \\
g24 & 2000 & 19990 & \bfseries  13337 & --- &  --- &  17905 &  --- &  17905 &  --- &  18070 \\
g25 & 2000 & 19990 & \bfseries  13340 & --- &  --- &  17993 &  --- &  17993 &  --- &  18049 \\
g26 & 2000 & 19990 & \bfseries  13328 & --- &  --- &  17744 &  --- &  18236 &  --- &  17744 \\
g27 & 2000 & 19990 & \bfseries  3341 & --- &  --- &  8273 &  --- &  8394 &  --- &  8273 \\
g28 & 2000 & 19990 & \bfseries  3298 & --- &  --- &  7505 &  --- &  8194 &  --- &  7505 \\
g29 & 2000 & 19990 & \bfseries  3405 & --- &  --- &  7594 &  --- &  7594 &  --- &  8382 \\
g30 & 2000 & 19990 & \bfseries  3413 & --- &  --- &  8033 &  --- &  8033 &  --- &  8354 \\
g31 & 2000 & 19990 & \bfseries  3310 & --- &  --- &  7688 &  --- &  8253 &  --- &  7688 \\
g32 & 2000 & 4000 & \bfseries  1410 & --- &  --- & \bfseries  1410 &  --- & \bfseries  1410 &  --- & \bfseries  1410 \\
g33 & 2000 & 4000 & \bfseries  1382 & --- &  --- & \bfseries  1382 &  --- & \bfseries  1382 &  --- & \bfseries  1382 \\
g34 & 2000 & 4000 & \bfseries  1384 & --- &  --- & \bfseries  1384 &  --- & \bfseries  1384 &  --- & \bfseries  1384 \\
g35 & 2000 & 11778 & \bfseries  7687 & --- &  --- &  8985 &  --- &  9242 &  --- &  8985 \\
g36 & 2000 & 11766 & \bfseries  7680 & --- &  --- &  9125 &  --- &  9199 &  --- &  9125 \\
g37 & 2000 & 11785 & \bfseries  7691 & --- &  --- &  9056 &  --- &  9056 &  --- &  9265 \\
g38 & 2000 & 11779 & \bfseries  7688 & --- &  --- &  8923 &  --- &  8923 &  --- &  9130 \\
g39 & 2000 & 11778 & \bfseries  2408 & --- &  --- &  3214 &  --- &  3254 &  --- &  3214 \\
g40 & 2000 & 11766 & \bfseries  2400 & --- &  --- &  3157 &  --- &  3157 &  --- &  3218 \\
g41 & 2000 & 11785 & \bfseries  2405 & --- &  --- &  3196 &  --- &  3196 &  --- &  3199 \\
g42 & 2000 & 11779 & \bfseries  2481 & --- &  --- &  3228 &  --- &  3228 &  --- &  3229 \\
g43 & 1000 & 9990 & \bfseries  6660 & --- &  --- &  8055 &  --- &  8264 &  --- &  8055 \\
g44 & 1000 & 9990 & \bfseries  6650 & --- &  --- &  8228 &  --- &  8228 &  --- &  8287 \\
g45 & 1000 & 9990 & \bfseries  6654 & --- &  --- &  8146 &  --- &  8182 &  --- &  8146 \\
g46 & 1000 & 9990 & \bfseries  6649 & --- &  --- &  8148 &  --- &  8168 &  --- &  8148 \\
g47 & 1000 & 9990 & \bfseries  6657 & --- &  --- &  8075 &  --- &  8146 &  --- &  8075 \\
g48 & 3000 & 6000 & \bfseries  6000 & --- &  --- & \bfseries  6000 &  --- & \bfseries  6000 &  --- & \bfseries  6000 \\
g49 & 3000 & 6000 & \bfseries  6000 & --- &  --- & \bfseries  6000 &  --- & \bfseries  6000 &  --- & \bfseries  6000 \\
g50 & 3000 & 6000 & \bfseries  5880 & --- &  --- & \bfseries  5880 &  --- & \bfseries  5880 &  --- & \bfseries  5880 \\
g51 & 1000 & 5909 & \bfseries  3848 & --- &  --- &  3992 &  --- &  4160 &  --- &  3992 \\
g52 & 1000 & 5916 & \bfseries  3851 & --- &  --- &  4095 &  --- &  4111 &  --- &  4095 \\
g53 & 1000 & 5914 & \bfseries  3850 & --- &  --- &  3971 &  --- &  3971 &  --- &  4090 \\
g54 & 1000 & 5916 & \bfseries  3852 & --- &  --- &  4073 &  --- &  4073 &  --- &  4082 \\
g55 & 5000 & 12498 & \bfseries  10299 & --- &  --- &  11692 &  --- &  11692 &  --- &  11755 \\
g56 & 5000 & 12498 & \bfseries  4016 & --- &  --- &  5378 &  --- &  5378 &  --- &  5418 \\
g57 & 5000 & 10000 & \bfseries  3494 & --- &  --- & \bfseries  3494 &  --- & \bfseries  3494 &  --- & \bfseries  3494 \\
g58 & 5000 & 29570 &  19289 & \bfseries 19290 & \bfseries  19290 &  24833 &  --- &  24833 & \bfseries  19290 &  25197 \\
g59 & 5000 & 29570 & \bfseries  6086 & --- &  --- &  10372 &  --- &  10372 &  --- &  10577 \\
g60 & 7000 & 17148 & \bfseries  14190 & --- &  --- &  16528 &  --- &  16547 &  --- &  16528 \\
g61 & 7000 & 17148 & \bfseries  5797 & --- &  --- &  8144 &  --- &  8144 &  --- &  8199 \\
g62 & 7000 & 14000 &  4868 & 4870 & \bfseries  4872 & \bfseries  4872 & \bfseries  4872 & \bfseries  4872 & \bfseries  4872 & \bfseries  4872 \\
g63 & 7000 & 41459 &  27033 & \bfseries 27037 &  --- &  35046 &  --- &  35046 &  --- &  35703 \\
g64 & 7000 & 41459 &  8747 & \bfseries 8748 &  --- &  15137 &  --- &  15137 &  --- &  15929 \\
g65 & 8000 & 16000 &  5560 & --- & \bfseries  5562 &  5568 &  --- &  5568 & \bfseries  5562 &  5568 \\
g66 & 9000 & 18000 &  6360 & --- & \bfseries  6364 &  6368 & \bfseries  6364 &  6368 & \bfseries  6364 &  6369 \\
g67 & 10000 & 20000 &  6942 & --- & \bfseries  6948 &  6952 &  --- &  6957 & \bfseries  6948 &  6952 \\
g70 & 10000 & 9999 &  9544 & 9551 & \bfseries  9575 &  9714 &  --- &  9714 & \bfseries  9575 &  9723 \\
g72 & 10000 & 20000 &  6998 & --- & \bfseries  7004 &  7013 & \bfseries  7004 &  7013 &  7002 &  7014 \\
g77 & 14000 & 28000 & \bfseries  9928 & --- &  --- &  9948 &  --- &  9948 &  --- &  9951 \\
g81 & 20000 & 40000 &  14036 & --- & \bfseries  14044 &  14078 &  --- &  14078 & \bfseries  14044 &  14080 \\
\end{tabular}
}

%% file: tab3.tex
\newcommand{\tabThree}{
\begin{tabular}{lrr||r|r||rr||rr||rr}
data &n&m& MOH & LS & ILP & UB& ILP$_{1}$ & UB$_{1}$& ILP$_{2}$ & UB$_{2}$  \\\hline
g1 & 800 & 19176 & \bfseries  15165 & --- &  --- &  19148 &  --- &  19148 &  --- &  19159 \\
g2 & 800 & 19176 & \bfseries  15172 & --- &  --- &  19160 &  --- &  19160 &  --- &  19160 \\
g3 & 800 & 19176 & \bfseries  15173 & --- &  --- &  19134 &  --- &  19134 &  --- &  19158 \\
g4 & 800 & 19176 & \bfseries  15184 & --- &  --- &  19117 &  --- &  19135 &  --- &  19117 \\
g5 & 800 & 19176 & \bfseries  15193 & --- &  --- &  19146 &  --- &  19146 &  --- &  19164 \\
g6 & 800 & 19176 & \bfseries  2632 & --- &  --- &  9441 &  --- &  9453 &  --- &  9441 \\
g7 & 800 & 19176 & \bfseries  2409 & --- &  --- &  9242 &  --- &  9253 &  --- &  9242 \\
g8 & 800 & 19176 & \bfseries  2428 & --- &  --- &  9228 &  --- &  9228 &  --- &  9230 \\
g9 & 800 & 19176 & \bfseries  2478 & --- &  --- &  9261 &  --- &  9275 &  --- &  9261 \\
g10 & 800 & 19176 & \bfseries  2407 & --- &  --- &  9233 &  --- &  9233 &  --- &  9267 \\
g11 & 800 & 1600 &  669 & --- & \bfseries  671 & \bfseries  671 & \bfseries  671 & \bfseries  671 & \bfseries  671 &  674 \\
g12 & 800 & 1600 &  660 & 661 & \bfseries  663 & \bfseries  663 & \bfseries  663 & \bfseries  663 & \bfseries  663 & \bfseries  663 \\
g13 & 800 & 1600 &  686 & 687 & \bfseries  688 & \bfseries  688 & \bfseries  688 & \bfseries  688 & \bfseries  688 & \bfseries  688 \\
g14 & 800 & 4694 & \bfseries  4012 & --- &  --- &  4497 &  --- &  4497 &  --- &  4510 \\
g15 & 800 & 4661 &  3984 & \bfseries 3985 & \bfseries  3985 &  4442 &  --- &  4442 & \bfseries  3985 &  4474 \\
g16 & 800 & 4672 & \bfseries  3990 & --- &  --- &  4458 &  --- &  4465 &  --- &  4458 \\
g17 & 800 & 4667 & \bfseries  3983 & --- &  --- &  4413 &  --- &  4413 &  --- &  4457 \\
g18 & 800 & 4694 & \bfseries  1207 & --- &  --- &  1962 &  --- &  1962 &  --- &  1991 \\
g19 & 800 & 4661 & \bfseries  1081 & --- &  --- &  1833 &  --- &  1859 &  --- &  1833 \\
g20 & 800 & 4672 & \bfseries  1122 & --- &  --- &  1888 &  --- &  1888 &  --- &  1889 \\
g21 & 800 & 4667 & \bfseries  1109 & --- &  --- &  1827 &  --- &  1827 &  --- &  1875 \\
g22 & 2000 & 19990 & \bfseries  17167 & --- &  --- &  19989 &  --- &  19989 &  --- &  19989 \\
g24 & 2000 & 19990 &  17162 & \bfseries 17163 &  --- &  19989 &  --- &  19989 &  --- &  19989 \\
g25 & 2000 & 19990 &  17163 & \bfseries 17164 &  --- &  19989 &  --- &  19989 &  --- &  19989 \\
g26 & 2000 & 19990 &  17154 & \bfseries 17155 &  --- &  19989 &  --- &  19990 &  --- &  19989 \\
g27 & 2000 & 19990 &  4020 & \bfseries 4021 &  --- &  9840 &  --- &  9841 &  --- &  9840 \\
g28 & 2000 & 19990 &  3973 & \bfseries 3975 &  --- &  9822 &  --- &  9827 &  --- &  9822 \\
g29 & 2000 & 19990 & \bfseries  4106 & --- &  --- &  9947 &  --- &  9948 &  --- &  9947 \\
g30 & 2000 & 19990 & \bfseries  4117 & --- &  --- &  9929 &  --- &  9929 &  --- &  9933 \\
g31 & 2000 & 19990 &  4003 & \bfseries 4005 &  --- &  9776 &  --- &  9861 &  --- &  9776 \\
g32 & 2000 & 4000 &  1653 & 1658 & \bfseries  1666 &  1668 & \bfseries  1666 &  1668 &  1664 &  1670 \\
g33 & 2000 & 4000 &  1625 & 1628 & \bfseries  1636 &  1640 & \bfseries  1636 &  1640 & \bfseries  1636 &  1640 \\
g34 & 2000 & 4000 &  1607 & 1609 & \bfseries  1616 &  1617 & \bfseries  1616 &  1617 &  1615 &  1618 \\
g35 & 2000 & 11778 &  10046 & \bfseries 10048 &  --- &  11711 &  --- &  11711 &  --- &  11714 \\
g36 & 2000 & 11766 & \bfseries  10039 & --- &  --- &  11702 &  --- &  11702 &  --- &  11703 \\
g37 & 2000 & 11785 &  10052 & \bfseries 10053 & \bfseries  10053 &  11691 &  --- &  11753 & \bfseries  10053 &  11691 \\
g38 & 2000 & 11779 & \bfseries  10040 & --- &  --- &  11703 &  --- &  11745 &  --- &  11703 \\
g39 & 2000 & 11778 & \bfseries  2903 & --- &  --- &  5457 &  --- &  5457 &  --- &  5551 \\
g40 & 2000 & 11766 &  2870 & \bfseries 2871 &  --- &  5471 &  --- &  5471 &  --- &  5471 \\
g41 & 2000 & 11785 &  2887 & \bfseries 2888 &  --- &  5452 &  --- &  5472 &  --- &  5452 \\
g42 & 2000 & 11779 & \bfseries  2980 & --- &  --- &  5551 &  --- &  5567 &  --- &  5551 \\
g43 & 1000 & 9990 & \bfseries  8573 & --- &  --- &  9985 &  --- &  9985 &  --- &  9988 \\
g44 & 1000 & 9990 & \bfseries  8571 & --- &  --- &  9957 &  --- &  9957 &  --- &  9980 \\
g45 & 1000 & 9990 & \bfseries  8566 & --- &  --- &  9983 &  --- &  9983 &  --- &  9986 \\
g46 & 1000 & 9990 & \bfseries  8568 & --- &  --- &  9983 &  --- &  9985 &  --- &  9983 \\
g47 & 1000 & 9990 & \bfseries  8572 & --- &  --- &  9966 &  --- &  9966 &  --- &  9983 \\
g48 & 3000 & 6000 & \bfseries  6000 & --- &  --- & \bfseries  6000 &  --- & \bfseries  6000 &  --- & \bfseries  6000 \\
g49 & 3000 & 6000 & \bfseries  6000 & --- &  --- & \bfseries  6000 &  --- & \bfseries  6000 &  --- & \bfseries  6000 \\
g50 & 3000 & 6000 & \bfseries  6000 & --- &  --- & \bfseries  6000 &  --- & \bfseries  6000 &  --- & \bfseries  6000 \\
g51 & 1000 & 5909 & \bfseries  5037 & --- &  --- &  5708 &  --- &  5712 &  --- &  5708 \\
g52 & 1000 & 5916 & \bfseries  5040 & --- &  --- &  5703 &  --- &  5703 &  --- &  5726 \\
g53 & 1000 & 5914 & \bfseries  5039 & --- &  --- &  5694 &  --- &  5694 &  --- &  5746 \\
g54 & 1000 & 5916 & \bfseries  5036 & --- &  --- &  5667 &  --- &  5667 &  --- &  5682 \\
g55 & 5000 & 12498 &  12427 & 12429 & \bfseries  12432 &  12498 &  --- &  12498 & \bfseries  12432 &  12498 \\
g56 & 5000 & 12498 &  4755 & \bfseries 4757 &  --- &  6157 &  --- &  6157 &  --- &  6176 \\
g57 & 5000 & 10000 &  4080 & 4092 & \bfseries  4103 &  4154 &  --- &  4154 & \bfseries  4103 &  4176 \\
g58 & 5000 & 29570 & \bfseries  25195 & --- &  --- &  29556 &  --- &  29556 &  --- &  29560 \\
g59 & 5000 & 29570 &  7274 & \bfseries 7276 &  --- &  14673 &  --- &  14673 &  --- &  14678 \\
g60 & 7000 & 17148 & \bfseries  17075 & --- &  --- &  17148 &  --- &  17148 &  --- &  17148 \\
g61 & 7000 & 17148 &  6858 & \bfseries 6861 &  --- &  8728 &  --- &  8728 &  --- &  8735 \\
g62 & 7000 & 14000 &  5686 & \bfseries 5710 &  5706 &  5981 &  --- &  6033 &  5706 &  5981 \\
g63 & 7000 & 41459 &  35315 & \bfseries 35318 &  --- &  41420 &  --- &  41420 &  --- &  41435 \\
g64 & 7000 & 41459 &  10429 & \bfseries 10437 &  --- &  20713 &  --- &  20747 &  --- &  20713 \\
g65 & 8000 & 16000 &  6489 & 6512 & \bfseries  6535 &  6711 &  --- &  6711 & \bfseries  6535 &  6970 \\
g66 & 9000 & 18000 &  7414 & 7442 & \bfseries  7443 &  7843 &  --- &  7843 & \bfseries  7443 &  8246 \\
g67 & 10000 & 20000 &  8088 & 8116 & \bfseries  8141 &  9080 &  --- &  9089 & \bfseries  8141 &  9080 \\
g70 & 10000 & 9999 & \bfseries  9999 & --- &  --- & \bfseries  9999 &  --- & \bfseries  9999 &  --- & \bfseries  9999 \\
g72 & 10000 & 20000 &  8190 & 8224 & \bfseries  8244 &  9166 &  --- &  9243 & \bfseries  8244 &  9166 \\
g77 & 14000 & 28000 &  11579 & \bfseries 11632 &  11619 &  13101 &  --- &  13101 &  11619 &  13104 \\
g81 & 20000 & 40000 &  16326 & \bfseries 16392 &  16374 &  18337 &  --- &  18515 &  16374 &  18337 \\
\end{tabular}
}

%% file: tab4.tex
\newcommand{\tabFour}{
\begin{tabular}{lrr||r|r||rr||rr||rr}
data &n&m& MOH & LS & ILP & UB& ILP$_{1}$ & UB$_{1}$& ILP$_{2}$ & UB$_{2}$  \\\hline
g1 & 800 & 19176 & \bfseries  16803 & --- &  --- &  19176 &  --- &  19176 &  --- &  19176 \\
g2 & 800 & 19176 & \bfseries  16809 & --- &  --- &  19176 &  --- &  19176 &  --- &  19176 \\
g3 & 800 & 19176 & \bfseries  16806 & --- &  --- &  19176 &  --- &  19176 &  --- &  19176 \\
g4 & 800 & 19176 & \bfseries  16814 & --- &  --- &  19176 &  --- &  19176 &  --- &  19176 \\
g5 & 800 & 19176 & \bfseries  16816 & --- &  --- &  19176 &  --- &  19176 &  --- &  19176 \\
g6 & 800 & 19176 & \bfseries  2751 & --- &  --- &  9544 &  --- &  9544 &  --- &  9583 \\
g7 & 800 & 19176 & \bfseries  2515 & --- &  --- &  9343 &  --- &  9430 &  --- &  9343 \\
g8 & 800 & 19176 & \bfseries  2525 & --- &  --- &  9397 &  --- &  9423 &  --- &  9397 \\
g9 & 800 & 19176 & \bfseries  2585 & --- &  --- &  9410 &  --- &  9410 &  --- &  9477 \\
g10 & 800 & 19176 & \bfseries  2510 & --- &  --- &  9380 &  --- &  9429 &  --- &  9380 \\
g11 & 800 & 1600 & \bfseries  677 & --- &  --- & \bfseries  677 &  --- & \bfseries  677 &  --- & \bfseries  677 \\
g12 & 800 & 1600 &  664 & --- & \bfseries  665 & \bfseries  665 & \bfseries  665 & \bfseries  665 & \bfseries  665 & \bfseries  665 \\
g13 & 800 & 1600 & \bfseries  690 & --- &  --- & \bfseries  690 &  --- & \bfseries  690 &  --- & \bfseries  690 \\
g14 & 800 & 4694 & \bfseries  4440 & --- &  --- &  4670 &  --- &  4671 &  --- &  4670 \\
g15 & 800 & 4661 & \bfseries  4406 & --- &  --- &  4622 &  --- &  4622 &  --- &  4644 \\
g16 & 800 & 4672 & \bfseries  4415 & --- &  --- &  4630 &  --- &  4635 &  --- &  4630 \\
g17 & 800 & 4667 & \bfseries  4411 & --- &  --- &  4625 &  --- &  4636 &  --- &  4625 \\
g18 & 800 & 4694 &  1261 & \bfseries 1262 & \bfseries  1262 &  2122 &  --- &  2140 & \bfseries  1262 &  2122 \\
g19 & 800 & 4661 & \bfseries  1121 & --- &  --- &  2045 &  --- &  2050 &  --- &  2045 \\
g20 & 800 & 4672 & \bfseries  1168 & --- &  --- &  2049 &  --- &  2049 &  --- &  2074 \\
g21 & 800 & 4667 & \bfseries  1155 & --- &  --- &  2023 &  --- &  2052 &  --- &  2023 \\
g22 & 2000 & 19990 & \bfseries  18776 & --- &  --- &  19990 &  --- &  19990 &  --- &  19990 \\
g24 & 2000 & 19990 &  18769 & \bfseries 18772 &  --- &  19990 &  --- &  19990 &  --- &  19990 \\
g25 & 2000 & 19990 &  18775 & \bfseries 18776 &  --- &  19990 &  --- &  19990 &  --- &  19990 \\
g26 & 2000 & 19990 &  18767 & \bfseries 18770 &  --- &  19990 &  --- &  19990 &  --- &  19990 \\
g27 & 2000 & 19990 &  4201 & \bfseries 4202 &  --- &  9928 &  --- &  9951 &  --- &  9928 \\
g28 & 2000 & 19990 &  4150 & \bfseries 4157 &  --- &  9888 &  --- &  9888 &  --- &  9919 \\
g29 & 2000 & 19990 &  4293 & \bfseries 4294 &  --- &  10010 &  --- &  10022 &  --- &  10010 \\
g30 & 2000 & 19990 &  4305 & \bfseries 4308 &  --- &  10019 &  --- &  10019 &  --- &  10024 \\
g31 & 2000 & 19990 &  4171 & \bfseries 4176 &  --- &  9910 &  --- &  9914 &  --- &  9910 \\
g32 & 2000 & 4000 &  1669 & 1671 & \bfseries  1679 & \bfseries  1679 & \bfseries  1679 & \bfseries  1679 & \bfseries  1679 & \bfseries  1679 \\
g33 & 2000 & 4000 &  1638 & 1640 & \bfseries  1644 & \bfseries  1644 & \bfseries  1644 & \bfseries  1644 & \bfseries  1644 & \bfseries  1644 \\
g34 & 2000 & 4000 &  1616 & 1617 & \bfseries  1623 & \bfseries  1623 & \bfseries  1623 & \bfseries  1623 & \bfseries  1623 &  1625 \\
g35 & 2000 & 11778 & \bfseries  11111 & --- &  --- &  11775 &  --- &  11776 &  --- &  11775 \\
g36 & 2000 & 11766 &  11108 & --- & \bfseries  11109 &  11763 &  --- &  11763 & \bfseries  11109 &  11764 \\
g37 & 2000 & 11785 &  11117 & \bfseries 11118 &  --- &  11785 &  --- &  11785 &  --- &  11785 \\
g38 & 2000 & 11779 &  11108 & \bfseries 11109 &  --- &  11778 &  --- &  11778 &  --- &  11778 \\
g39 & 2000 & 11778 &  3006 & \bfseries 3007 &  --- &  5736 &  --- &  5794 &  --- &  5736 \\
g40 & 2000 & 11766 &  2976 & \bfseries 2978 &  --- &  5665 &  --- &  5669 &  --- &  5665 \\
g41 & 2000 & 11785 &  2983 & \bfseries 2986 &  2984 &  5751 &  --- &  5751 &  2984 &  5758 \\
g42 & 2000 & 11779 &  3092 & \bfseries 3095 &  --- &  5787 &  --- &  5800 &  --- &  5787 \\
g43 & 1000 & 9990 &  9376 & \bfseries 9377 & \bfseries  9377 &  9990 &  --- &  9990 & \bfseries  9377 &  9990 \\
g44 & 1000 & 9990 & \bfseries  9379 & --- &  --- &  9990 &  --- &  9990 &  --- &  9990 \\
g45 & 1000 & 9990 &  9376 & \bfseries 9377 & \bfseries  9377 &  9990 &  --- &  9990 & \bfseries  9377 &  9990 \\
g46 & 1000 & 9990 & \bfseries  9378 & --- &  --- &  9990 &  --- &  9990 &  --- &  9990 \\
g47 & 1000 & 9990 & \bfseries  9381 & --- &  --- &  9990 &  --- &  9990 &  --- &  9990 \\
g48 & 3000 & 6000 & \bfseries  6000 & --- &  --- & \bfseries  6000 &  --- & \bfseries  6000 &  --- & \bfseries  6000 \\
g49 & 3000 & 6000 & \bfseries  6000 & --- &  --- & \bfseries  6000 &  --- & \bfseries  6000 &  --- & \bfseries  6000 \\
g50 & 3000 & 6000 & \bfseries  6000 & --- &  --- & \bfseries  6000 &  --- & \bfseries  6000 &  --- & \bfseries  6000 \\
g51 & 1000 & 5909 &  5571 & \bfseries 5572 & \bfseries  5572 &  5871 &  --- &  5881 & \bfseries  5572 &  5871 \\
g52 & 1000 & 5916 & \bfseries  5584 & --- &  --- &  5891 &  --- &  5891 &  --- &  5891 \\
g53 & 1000 & 5914 & \bfseries  5574 & --- &  --- &  5887 &  --- &  5887 &  --- &  5888 \\
g54 & 1000 & 5916 & \bfseries  5579 & --- &  --- &  5889 &  --- &  5889 &  --- &  5889 \\
g55 & 5000 & 12498 & \bfseries  12498 & --- &  --- & \bfseries  12498 &  --- & \bfseries  12498 &  --- & \bfseries  12498 \\
g56 & 5000 & 12498 &  4931 & \bfseries 4935 &  --- &  6213 &  --- &  6213 &  --- &  6213 \\
g57 & 5000 & 10000 &  4112 & 4132 & \bfseries  4145 &  4220 &  4141 &  4305 & \bfseries  4145 &  4220 \\
g58 & 5000 & 29570 & \bfseries  27885 & --- &  --- &  29569 &  --- &  29569 &  --- &  29570 \\
g59 & 5000 & 29570 &  7539 & \bfseries 7546 &  --- &  14731 &  --- &  14731 &  --- &  14731 \\
g60 & 7000 & 17148 & \bfseries  17148 & --- &  --- & \bfseries  17148 &  --- & \bfseries  17148 &  --- & \bfseries  17148 \\
g61 & 7000 & 17148 &  7110 & \bfseries 7114 &  --- &  8748 &  --- &  8751 &  --- &  8748 \\
g62 & 7000 & 14000 &  5743 & 5758 & \bfseries  5788 &  6534 &  5774 &  6534 & \bfseries  5788 &  6541 \\
g63 & 7000 & 41459 &  39083 & \bfseries 39089 &  --- &  41459 &  --- &  41459 &  --- &  41459 \\
g64 & 7000 & 41459 &  10814 & \bfseries 10819 &  --- &  20775 &  --- &  20775 &  --- &  20792 \\
g65 & 8000 & 16000 &  6534 & 6561 & \bfseries  6579 &  7256 &  6573 &  7256 & \bfseries  6579 &  7349 \\
g66 & 9000 & 18000 &  7474 & 7495 & \bfseries  7522 &  8497 &  7505 &  8497 & \bfseries  7522 &  8500 \\
g67 & 10000 & 20000 &  8155 & 8185 & \bfseries  8220 &  9299 &  --- &  9299 & \bfseries  8220 &  9303 \\
g70 & 10000 & 9999 & \bfseries  9999 & --- &  --- & \bfseries  9999 &  --- & \bfseries  9999 &  --- & \bfseries  9999 \\
g72 & 10000 & 20000 &  8264 & 8296 & \bfseries  8337 &  9357 &  --- &  9357 & \bfseries  8337 &  9376 \\
g77 & 14000 & 28000 &  11674 & \bfseries 11712 &  11691 &  13296 &  --- &  13296 &  11691 &  13455 \\
g81 & 20000 & 40000 &  16470 & \bfseries 16525 &  16485 &  19088 &  --- &  19088 &  16485 &  19580 \\
\end{tabular}
}

%% file: tableImpSmall.tex
\newcommand{\improvementType}{
\begin{tabular}{r|cccccc}
&improvable&I$_1$&I$_2$&I$_3$&LS&ILP\\\hline 
unit $ c = 2$ & 31&2  &1  &0   & 3 & 2\\
unit $ c = 3$ & 30&8  &0  &0   & 8 & 3\\
unit $ c = 4$ & 28&5  &3  &1   & 9 & 4\\
\hline 
signed $ c = 2$ & 29 & 1   & 1   & 0   & 2 & 6\\
signed $ c = 3$ & 36 & 19   & 2   & 1   & 22 & 14\\
signed $ c = 4$ & 34 & 20   & 5   & 0   & 25 & 14\\
\hline
sum 
 & 188 & 55 & 12 & 2 & 69 & 43\end{tabular}
}

\newcommand{\improvementC}{
\begin{tabular}{r|cccccc}
&improvable&I$_1$&I$_2$&I$_3$&LS&ILP\\\hline 
\hline unit $ c = 2$ & 31 & 2   & 1   & 0   & 3 & 2\\
signed $ c = 2$ & 29 & 1   & 1   & 0   & 2 & 6\\
\hline unit $ c = 3$ & 30 & 8   & 0   & 0   & 8 & 3\\
signed $ c = 3$ & 36 & 19   & 2   & 1   & 22 & 14\\
\hline unit $ c = 4$ & 28 & 5   & 3   & 1   & 9 & 4\\
signed $ c = 4$ & 34 & 20   & 5   & 0   & 25 & 14\\
\hline
sum 
 & 188 & 55 & 12 & 2 & 69 & 43\end{tabular}
}

\newcommand{\firstType}{
\begin{tabular}{r|rrrrrrrrr}
& 2& 3& 4& 5& 6& 7& 8& 9& 10\\\hline 
unit $ c = 2$ & 1& 0& 0& 2& 0& 0& 0& 0& 0\\
unit $ c = 3$ & 4& 3& 1& 0& 0& 0& 0& 0& 0\\
unit $ c = 4$ & 5& 0& 1& 2& 1& 0& 0& 0& 0\\
\hline
signed $ c = 2$ & 0& 0& 1& 0& 0& 1& 0& 0& 0\\
signed $ c = 3$ & 7& 10& 1& 3& 0& 0& 0& 0& 1\\
signed $ c = 4$ & 3& 13& 5& 2& 2& 0& 0& 0& 0\\
\hline
sum 
 & 20 & 26 & 9 & 9 & 3 & 1 & 0 & 0 & 1\end{tabular}
}

\newcommand{\firstC}{
\begin{tabular}{r|rrrrrrrrr}
& 2& 3& 4& 5& 6& 7& 8& 9& 10\\
\hline unit $ c = 2$ & 1& 0& 0& 2& 0& 0& 0& 0& 0\\
signed $ c = 2$ & 0& 0& 1& 0& 0& 1& 0& 0& 0\\
\hline unit $ c = 3$ & 4& 3& 1& 0& 0& 0& 0& 0& 0\\
signed $ c = 3$ & 7& 10& 1& 3& 0& 0& 0& 0& 1\\
\hline unit $ c = 4$ & 5& 0& 1& 2& 1& 0& 0& 0& 0\\
signed $ c = 4$ & 3& 13& 5& 2& 2& 0& 0& 0& 0\\
\hline
sum 
 & 20 & 26 & 9 & 9 & 3 & 1 & 0 & 0 & 1\end{tabular}
}

\newcommand{\firstILPBadType}{
\begin{tabular}{r|rrrrrrrrr}
& 2& 3& 4& 5& 6\\\hline 
unit $ c = 2$ & 1& 0& 0& 0& 0\\
unit $ c = 3$ & 3& 2& 0& 0& 0\\
unit $ c = 4$ & 3& 0& 1& 1& 1\\
\hline
signed $ c = 2$ & 0& 0& 1& 0& 0\\
signed $ c = 3$ & 2& 4& 1& 2& 0\\
signed $ c = 4$ & 3& 3& 4& 1& 1\\
\hline
sum 
 & 12 & 9 & 7 & 4 & 2\end{tabular}
}

\newcommand{\firstILPBadC}{
\begin{tabular}{r|rrrrrrrrr}
& 2& 3& 4& 5& 6\\
\hline unit $ c = 2$ & 1& 0& 0& 0& 0\\
signed $ c = 2$ & 0& 0& 1& 0& 0\\
\hline unit $ c = 3$ & 3& 2& 0& 0& 0\\
signed $ c = 3$ & 2& 4& 1& 2& 0\\
\hline unit $ c = 4$ & 3& 0& 1& 1& 1\\
signed $ c = 4$ & 3& 3& 4& 1& 1\\
\hline
sum 
 & 12 & 9 & 7 & 4 & 2\end{tabular}
}

\newcommand{\largestType}{
\begin{tabular}{r|rrrrrrrrrrrr}
& 2& 3& 4& 5& 6& 7& 8& 9& 10& 11& 12& 13\\\hline 
unit $ c = 2$ & 0& 0& 0& 2& 0& 0& 0& 0& 0& 0& 0& 1\\
unit $ c = 3$ & 3& 1& 3& 0& 0& 1& 0& 0& 0& 0& 0& 0\\
unit $ c = 4$ & 2& 0& 3& 3& 1& 0& 0& 0& 0& 0& 0& 0\\
\hline
signed $ c = 2$ & 0& 0& 1& 0& 0& 1& 0& 0& 0& 0& 0& 0\\
signed $ c = 3$ & 0& 4& 0& 2& 3& 1& 1& 8& 3& 0& 0& 0\\
signed $ c = 4$ & 1& 2& 4& 3& 6& 5& 4& 0& 0& 0& 0& 0\\
\hline
sum 
 & 6 & 7 & 11 & 10 & 10 & 8 & 5 & 8 & 3 & 0 & 0 & 1\end{tabular}
}

\newcommand{\largestC}{
\begin{tabular}{r|rrrrrrrrrrrr}
& 2& 3& 4& 5& 6& 7& 8& 9& 10& 11& 12& 13\\
\hline unit $ c = 2$ & 0& 0& 0& 2& 0& 0& 0& 0& 0& 0& 0& 1\\
signed $ c = 2$ & 0& 0& 1& 0& 0& 1& 0& 0& 0& 0& 0& 0\\
\hline unit $ c = 3$ & 3& 1& 3& 0& 0& 1& 0& 0& 0& 0& 0& 0\\
signed $ c = 3$ & 0& 4& 0& 2& 3& 1& 1& 8& 3& 0& 0& 0\\
\hline unit $ c = 4$ & 2& 0& 3& 3& 1& 0& 0& 0& 0& 0& 0& 0\\
signed $ c = 4$ & 1& 2& 4& 3& 6& 5& 4& 0& 0& 0& 0& 0\\
\hline
sum 
 & 6 & 7 & 11 & 10 & 10 & 8 & 5 & 8 & 3 & 0 & 0 & 1\end{tabular}
}